\pdfoutput=1
\RequirePackage{ifpdf}
\ifpdf 
\documentclass[pdftex]{sigma}
\else
\documentclass{sigma}
\fi

\numberwithin{equation}{section}

\newtheorem{Theorem}{Theorem}[section]
\newtheorem{Proposition}[Theorem]{Proposition}

\def\bO{\bar 0}
\def\b1{\bar 1}
\def\lieg{\mathfrak{g}}

\begin{document}

\newcommand{\arXivNumber}{2009.00393}

\renewcommand{\thefootnote}{}

\renewcommand{\PaperNumber}{007}

\FirstPageHeading

\ShortArticleName{Harmonic Analysis in $d$-Dimensional Superconformal Field Theory}

\ArticleName{Harmonic Analysis in $\boldsymbol{d}$-Dimensional Superconformal\\ Field Theory\footnote{This paper is a~contribution to the Special Issue on Representation Theory and Integrable Systems in honor of Vitaly Tarasov on the 60th birthday and Alexander Varchenko on the 70th birthday. The full collection is available at \href{https://www.emis.de/journals/SIGMA/Tarasov-Varchenko.html}{https://www.emis.de/journals/SIGMA/Tarasov-Varchenko.html}}}

\Author{Ilija BURI\'C}

\AuthorNameForHeading{I.~Buri\'c}

\Address{DESY, Notkestra{\ss}e 85, D-22607 Hamburg, Germany}
\Email{\href{mailto:ilija.buric@desy.de}{ilija.buric@desy.de}}

\ArticleDates{Received September 02, 2020, in final form January 15, 2021; Published online January 25, 2021}

\Abstract{Superconformal blocks and crossing symmetry equations are among central ingredients in any superconformal field theory. We review the approach to these objects rooted in harmonic analysis on the superconformal group that was put forward in [\textit{J.~High Energy Phys.} \textbf{2020} (2020), no.~1, 159, 40~pages, arXiv:1904.04852] and [\textit{J.~High Energy Phys.} \textbf{2020} (2020), no.~10, 147, 44~pages, arXiv:2005.13547]. After lifting conformal four-point functions to functions on the superconformal group, we explain how to obtain compact expressions for crossing constraints and Casimir equations. The later allow to write superconformal blocks as finite sums of spinning bosonic blocks.}


\Classification{81R05; 81R12}


\renewcommand{\thefootnote}{\arabic{footnote}}
\setcounter{footnote}{0}

\section{Introduction}

Conformal field theories (CFTs) are a class of quantum field theories that are interesting for several reasons. On the one hand, they describe the critical behaviour of statistical mechanics systems such as the Ising model. Indeed, the identification of two-dimensional statistical systems with CFT minimal models, first suggested in~\cite{Belavin:1984vu}, was a celebrated early achievement in the field. For similar reasons, conformal theories classify universality classes of quantum field theories in the Wilsonian renormalisation group paradigm. On the other hand, CFTs also play a role in the description of physical systems that do not posses scale invariance, through certain ``dualities''. The most prominent of these is the conjectured AdS/CFT correspondence, according to which conformal field theories should be related to quantum theories of gravity.

An attractive feature of CFTs is that there exists a variety of non-perturbative methods for their study. This is especially true for theories with supersymmetry. Techniques that are being used include integrability, holography, chiral algebras, superconformal index calculations, supersymmetric localisation and the conformal bootstrap.

The last approach, one that we will be concerned with, relies on an axiomatisation of conformal theories that is based on two main assumptions. Firstly, the Hilbert space $\mathcal{H}$ of the theory is assumed to carry a unitary representation of the group $G$ of conformal transformations. Secondly, the space $\mathcal{H}$ is equipped with an algebraic structure called the operator product expansion (OPE). Roughly speaking, the OPE makes $\mathcal{H}$ into an algebra.

The decomposition of the Hilbert space into irreducible representations of $G$ and the structure constants of the operator product algebra define CFT data. The CFT data completely determines the theory, in the sense that it allows for the computation of all correlation functions. In the conformal bootstrap, one tries to constrain the CFT data from self-consistency and some basic physical requirements. The most constraining consistency conditions are the ones implied by the associativity of the operator product algebra. Usually, this condition is formulated as a~property of four-point correlation functions called the crossing symmetry. Therefore, from the mathematical perspective, the bootstrap is a classification programme for solutions of crossing symmetry equations. It was formulated in the 70s by three groups \cite{Ferrara:1973yt,Mack:1976pa,Polyakov:1974gs}, but little progress in dimensions higher than two was made until 2008 when it was realised in~\cite{Rattazzi:2008pe} that crossing equations can be efficiently studied numerically. Since then, there have been many significant advances, both in numerical and in analytical studies of bootstrap equations, exemplified by the precise determination of critical exponents in three-dimensional Ising model, \cite{ElShowk:2012ht,El-Showk:2014dwa,Kos:2016ysd}. Still, no exact solutions beyond the free theory are known in dimensions higher than two.

One among promising ideas in the bootstrap endeavour came in $\cite{Isachenkov:2016gim}$, where the authors realised that conformal partial waves, which capture the contribution of a single irreducible conformal representation to the correlation function, can in some cases be identified with wavefunctions of an integrable Schr\"odinger problem of Calogero--Sutherland type. This was explained in~\cite{Schomerus:2016epl} through harmonic analysis on the conformal group. It is the harmonic analysis approach to conformal theories that we will adopt in the present work.

In Section~\ref{section2}, we will review the construction of crossing equations, illustrating it on a simple example. The two ingredients entering the equations will be defined. These are conformal blocks and the so called crossing factors. We shall follow the influential approach of Dolan and Osborn who characterised the blocks as solutions to a set of Casimir differential equations, \cite{Dolan:2000ut,Dolan:2003hv}. We will then comment on the extent to which both conformal blocks and crossing factors are known (they depend on spacetime dimension, amount of supersymmetry and type of fields entering the four-point function). This material is very well known and is included in order to formulate clearly the problem that we wish to address in the rest.

Sections~\ref{section3}--\ref{section5} are devoted to recent constructions of \cite{Buric:2019rms,Buric:2020buk}. Section~\ref{section3} describes how superconformal four-point functions can be put in a correspondence with certain covariant functions on the superconformal group, termed the $K$-spherical functions. In order to achieve this, we will lift the fields of the theory to functions on the group.

The gain of this initial step lies in the fact that $K$-spherical functions can be studied by established methods of group theory. In Section~\ref{section4}, we shall use the Cartan decomposition of the superconformal group to construct the crossing symmetry equations. Before \cite{Buric:2020buk}, these equations have been constructed only in a limited number of cases. In Section~\ref{section5}, we turn to superconformal blocks. They will be characterised as eigenfunctions of the Laplace--Beltrami operator within the space of $K$-spherical functions. In bosonic theories,\footnote{Throughout the text we will use the word ``bosonic'' to refer to theories which do not posses supersymmetry. These theories can have both bosonic and fermionic fields.} this eigenvalue problem assumes the form of a matrix-valued Schr\"odinger equation that generalises the $BC_2$ Calogero--Sutherland system. This allows to express the blocks in many cases in terms of known special functions. In the supersymmetric setup, partial waves will be obtained from the bosonic ones through a~quantum mechanical perturbation theory that becomes exact at a~small finite order.

A recurring theme will be the fact that the group theoretic approach, among other benefits, allows to treat bosonic and supersymmetric conformal theories in a very similar manner. The latter are necessarily more involved in the kinematical aspects considered here, but we will see that the additional complications appear in a controlled way and can be systematically dealt with. On the other hand, once these difficulties are overcome, supersymmetric theories are certainly attractive to study due to the number of non-perturbative methods available for them that were mentioned above.

The work described here is a first step towards the study of crossing equations, only enabling for their formulation in a large variety of cases. Some remarkable progress in understanding the structure of solutions, and thereby conformal field theories, was made in recent years by considerations of the simplest correlator of identical scalar fields $\langle\varphi\varphi\varphi\varphi\rangle$, \cite{Alday:2016njk,Caron-Huot:2017vep,Fitzpatrick:2012yx, Komargodski:2012ek}. It should be possible to obtain much more information by studying crossing for various types of fields. In the concluding Section~\ref{section6}, we will comment on the future developments in this direction that we hope our results may lead to.

\section{Crossing symmetry equations}\label{section2}

This section may be understood as an extended introduction, where, to fix the ideas, we consider the problem that we want to address in the simplest setup. After reviewing its solution, we will go on to formulate what kind of generalisations will be considered in later sections.

To this end, let us consider a bosonic conformal field theory on $M = \mathbb{R}^d\cup\{\infty\}$. A correlation function $G_n(x_i)=\langle\mathcal{O}_1(x_1)\cdots \mathcal{O}_n(x_n)\rangle$ of local primary fields is a vector valued function $M^n\xrightarrow{} V=V_1\otimes\cdots \otimes V_n$ from $n$ copies of spacetime $M$ into a tensor product of finite dimensional vector spaces $V_i$.\footnote{Correlation functions are not defined when some of the insertion points coincide. We will assume throughout the text that the points are in general position where this is not the case.} The spaces~$V_i$ carry representations~$\rho_i$ of a subgroup $K = {\rm SO}(d)\times {\rm SO}(1,1)$ of the conformal group $G={\rm SO}(d+1,1)$, generated by rotations and dilations. Each representation~$\rho_i$ is specified by a conformal weight~$\Delta_i$ and a highest weight $\lambda_i$ for ${\rm SO}(d)$~-- the spin of the field. By conformal invariance, $G_n(x_i)$ satisfies a set of covariance conditions known as the Ward identities. This allows one to write it in the form
\begin{equation}\label{correlator}
 G_n(x_i) = \Omega(x_i) F(u_a),
\end{equation}
where $u_a$ are conformal invariants (cross ratios) constructed from points $x_i$ and $F$ is an arbitrary function that takes values in a vector space $W$ of dimension less than or equal to that of~$V$. The prefactor~$\Omega(x_i)$ ensures the correct behaviour under conformal transformations $x_i\mapsto g x_i$. Usually, $W$ is called the space of tensor structures and~$\Omega$ the tensor factor. The decomposition~(\ref{correlator}) is not unique, since one can redefine~$\Omega$ by multiplying it by an arbitrary function of cross ratios~$u_a$.

For the four-point function of identical scalars $\varphi$ with the conformal weight $\Delta_\varphi$ the usual choice is
\begin{equation}\label{4pt-function}
 G_4(x_i) = \frac{1}{x_{12}^{2\Delta_\varphi} x_{34}^{2\Delta_\varphi}} F(u,v),\qquad u = \frac{x_{12}^2 x_{34}^2}{x_{13}^2 x_{24}^2},\qquad v = \frac{x_{14}^2 x_{23}^2}{x_{12}^2 x_{34}^2} .
\end{equation}
Here, we use the notation $x_{ij} = x_i - x_j$. The number of conformal invariants in this case is two because, starting from four points $x_i$ in general position, one can use conformal transformations to map them to
\begin{equation}\label{points-choice}
 x_1\mapsto0,\qquad x_2\mapsto \frac{z_1+z_2}{2} e_1 + \frac{z_1 - z_2}{2i} e_2,\qquad x_3\mapsto e_1,\qquad x_4\mapsto\infty,
\end{equation}
where $\{e_i\}$ is a standard orthonormal basis of $\mathbb{R}^d$. The coordinates $(z_1,z_2)$ of the point $x_2$ are then related to the cross ratios $u$, $v$ through
\begin{equation*}
 z_1 z_2 = u,\qquad (1-z_1)(1-z_2) = v .
\end{equation*}
The conformal Lie algebra $\mathfrak{g} = \mathfrak{so}(d+1,1)$ is represented on the space of scalar fields on $M$ through differential operators
\begin{gather*}
 D = x^\mu\partial_\mu + \Delta_\varphi,\qquad P_\mu = \partial_\mu,\qquad M_{\mu\nu} = x_\nu \partial_\mu - x_\mu\partial_\nu,\\ K_\mu = - x^2\partial_\mu + 2 x_\mu x^\nu\partial_\nu + 2x_\mu \Delta_\varphi .
\end{gather*}
We shall put an additional index $i$ on the operators to mean that $x$ in the above formulas is the variable $x_i$ and partial derivatives are with respect to $x_i^\mu$. Furthermore, it is convenient to introduce another basis $\{L_{\alpha\beta}\}$ for $\mathfrak{g}$ by
\begin{gather*}
 L_{0\mu} = \tfrac12(P_\mu - K_\mu),\qquad L_{1\mu} = \tfrac12 (P_\mu + K_\mu),\qquad L_{01} = D,\\ L_{\mu\nu} = M_{\mu\nu}, \qquad \mu,\nu=2,\dots,d+1 .
\end{gather*}
The quadratic Casimir $C_2 = L_{\alpha\beta} L^{\alpha\beta}$ is a second order differential operator in the above representation. Let $4C^{(12)}_2 = -\big(L^1_{\alpha\beta}+L^2_{\alpha\beta}\big)\big(L^{1,\alpha\beta}+L^{2,\alpha\beta}\big)$. A conformal partial wave captures the contribution to the correlation function of one conformal family present in the OPE. As noticed by Dolan and Osborn, \cite{Dolan:2000ut}, waves may be characterised as eigenfunctions of~$C_2^{(12)}$. When acting on $G_4(x_i)$, the operator $C_2^{(12)}$ produces a function of the same product form~(\ref{4pt-function}). Therefore, the operator $\Delta_2=\Omega(x_i)^{-1} C_2^{(12)} \Omega(x_i)$ can be written as a differential operator in the cross ratios, called the Casimir differential operator. It is in fact simpler to write $\Delta_2$ in variables $z_i$, where it can be shown to take the form
\begin{equation*}
 \Delta_2 = D_{z_1} + D_{z_2} + (d-2)\frac{z_1 z_2}{z_1 - z_2}((1-z_1)\partial_{z_1} - (1-z_2)\partial_{z_2}) .
\end{equation*}
Here, the operator $D_x$ reads
\begin{equation*}
 D_x = x^2 (1-x) \partial^2_x - x^2\partial_x .
\end{equation*}
Therefore, in $d=2$ dimensions, the Casimir eigenvalue equations split into independent hypergeometric equations in~$z_1$ and $z_2$. With some additional work, it is possible to decouple the equations for any even~$d$. Conformal blocks $g_{\Delta,l}$ are eigenfunctions of $\Delta_2$ with eigenvalues $2\Delta(\Delta-d)+2l(l+d-2)$. In two and four dimensions they read
\begin{gather*}
  g_{\Delta,l}^{(2d)} = k_{\Delta+l}(z_1) k_{\Delta-l}(z_2) + k_{\Delta-l}(z_1) k_{\Delta+l}(z_2),\\
 g_{\Delta,l}^{(4d)} = \frac{z_1 z_2}{z_1 - z_2} \big( k_{\Delta+l}(z_1) k_{\Delta-l-2}(z_2) - k_{\Delta-l-2}(z_1) k_{\Delta+l}(z_2)\big),
\end{gather*}
where $k$ is given in terms of the hypergeometric function by
\begin{equation*}
 k_{2a}(x) = x^{a}\,{} _2F_1(a,a;2a;x) .
\end{equation*}
The correlation function can be expanded in conformal blocks. The coefficients in the expansion can be seen to be squares of the OPE coefficients, and are therefore positive real numbers. We denote them by $p_{\Delta,l}$ and write
\begin{equation*}
 G_4(x_i) = \frac{1}{x_{12}^{2\Delta_\varphi}x_{34}^{2\Delta_\varphi}} \sum_{\Delta,l} p_{\Delta,l} g_{\Delta,l}(u,v) .
\end{equation*}
Another property of correlators, which follows from Euclidean quantum field theory axioms, is invariance under permutations of the arguments $x_i$. Whereas such a condition may seem rather innocent, when combined with the decomposition~(\ref{correlator}), it leads to a non-trivial functional equation for $F$
\begin{equation}\label{crossing-eqn-scalars}
 \Omega(x_i) F(u_a) = \Omega(x_{\sigma(i)}) F(u'_a) .
\end{equation}
Here, $\sigma$ is any permutation in $S_4$ and $u'_a$ stand for cross ratios constructed out of permuted points. Any particular permutation is referred to as a {\it channel}. Taking $\sigma = (2 4)$ and $\sigma = (3 4)$ in the example above gives
\begin{equation*}
 F(u,v) = \left( \frac{u}{v} \right)^{\Delta_\varphi} F(v,u),\qquad F(u,v) = F(u/v,1/v).
\end{equation*}
The idea of conformal bootstrap is to substitute the conformal block decomposition for $F$ on both sides of these equations and try to find solutions with positive coefficients $p_{\Delta,l}$. The second of these equations is actually satisfied by each block of even spin, while for odd spins the two sides differ by a sign. Therefore, we learn that only operators of even spin appear in the decomposition of $G_4(x_i)$. The first equation, however, leads to a much more non-trivial condition on the $p_{\Delta,l}$
\begin{equation*}
 v^{\Delta_\varphi} \sum_{\Delta,l} p_{\Delta,l} g_{\Delta,l}(u,v) = u^{\Delta_\varphi} \sum_{\Delta,l} p_{\Delta,l} g_{\Delta,l}(v,u).
\end{equation*}
This is the basic construction that we want to carry out in more complicated situations. By that we mean that the fields in the correlation function will be allowed to carry arbitrary spins (in general, different from each other). In such a case the above procedure meets some difficulties. Firstly, although conformal blocks are still characterised as solutions to appropriate Casimir equations, it is often hard to identify them in the world of special functions. For bosonic theories, there exist efficient algorithmic procedures for computation of such spinning blocks, \cite{Echeverri:2015rwa,Echeverri:2016dun,Costa:2016hju,Costa:2011dw,Erramilli:2019njx,Fortin:2020ncr,Karateev:2017jgd, Penedones:2015aga,SimmonsDuffin:2012uy}. One systematic and efficient method is to start from scalar blocks and apply to them a set of differential operators known as weight-shifting operators (see also \cite{Dolan:2011dv,Hogervorst:2013sma,Isachenkov:2017qgn} for other investigations of bosonic conformal blocks). It is probably fair to say that the theory of superconformal blocks is considerably less developed than its bosonic counterpart.

Concerning the crossing symmetry equations, for spinning fields the factor $\Omega(x_i)$ becomes a~$n\times m$ matrix with $n = \dim V$ and $m=\dim W$. The ratio of two prefactors that correspond to different permutations $\sigma$ is replaced by an $m\times m$ matrix $M$ whose entries depend on $x_i$ only through cross ratios. This {\it crossing factor} $M$ have been derived in several cases, usually for low spins, in \cite{Costa:2011mg,Cuomo:2017wme,Dymarsky:2017yzx,Dymarsky:2017xzb,Kravchuk:2016qvl, Osborn:1993cr}.

In this work, we shall review two recent constructions, \cite{Buric:2019rms,Buric:2020buk}, that address the above two questions in turn. Our starting point will be to regard the space $M$ as a coset of the conformal group by a certain parabolic subgroup $P$. Any function on $M$ thus can be lifted to a function on the group that is (left) covariant with respect to $P$. As we will see, the lifted function is then nothing but a vector in a principal series representation of $G$.

\looseness=-1 In this way, the four point function is lifted to a function $F_4\colon G^4\xrightarrow{} V$ that may be regarded as a vector in the tensor product of four principal series representations. Ward identities satisfied by $G_4$ mean that $F_4$ is an invariant vector. Such invariant vectors will be shown to give rise to functions $F\colon G\xrightarrow{} V$, covariant under both the left and the right action of a group $K\subset G$. Whereas the group $P$ is the stabiliser of one point in~$M$, the group $K$ is the stabiliser of a~pair of points. A key ingredient from group theory that allows for this transformation is the so-called Cartan decomposition of the conformal group. The functions with covariance laws obeyed by~$F$ will be referred to as $K$-spherical. Thus, our first result can be stated as producing a 1-1 correspondence between solutions of Ward identities and $K$-spherical functions. It is given in equation~(\ref{magic-formula}).

The latter space of functions is somewhat better adopted to group theory than the former. This will allow us to find universal formulas for the crossing factor and Casimir equations. We will make use of the following observation in particular: while the prefactor $\Omega(x_i)$ is not conformally invariant, the ratio of prefactors in two different channels is. This crossing factor is thus defined initially as a function of $4d$ variables in $d$ spacetime dimensions, but only depends on them through two cross ratios. In fact the crossing factor is a much simpler object than $\Omega(x_i)$ and we will provide compact formulas for arbitrary spinning fields in Section~\ref{section4}. For simplicity, we described the constructions assuming that $G$ is a bosonic conformal group, but they will be carried for an arbitrary superconformal group. In this case, $M$ is the corresponding superspace.

As mentioned in the introduction, we will derive Casimir equations in Section~\ref{section5}. Their solutions will be constructed as finite sums of bosonic spinning conformal blocks. The latter, as noted above, have been the subject of many recent investigations and are well understood.

Throughout the paper, all constructions will be illustrated on the example of $\mathcal{N}=2$ superconformal symmetry in one dimension. The Lie superalgebra of this symmetry is $\mathfrak{sl}(2|1)$.

\section{Lift of correlation functions to the group}\label{section3}

In this section, we shall establish a correspondence between four-point functions in a superconformal field theory and certain covariant functions on the superconformal group that may be called $K$-spherical functions. The first subsection introduces the Weyl inversion~-- an element $w$ of the bosonic part of the superconformal group, which is closely related to the usual conformal inversion, but has the advantage of being well-defined for an arbitrary superconformal group. Next, the Bruhat decomposition of ${\rm Spin}(d+1,1)$ is review and its super-cousin defined.

These ingredients are used in the final subsection to map solutions of Ward identities satisfied by a four-point function to vector-valued covariant functions on $G$. Our starting point is the observation that the Ward identities can be written as (\ref{eq:G4Wardid}) using the Bruhat factors. Then a~$K$-spherical function $F$ is produced from a~solution~$G_4(x_i)$ in equations (\ref{eq:F4rightcov}) and (\ref{eq:FfromF4}). The space of $K$-spherical functions is defined in~(\ref{eq:covariance}). Finally, we show how to invert the process and recover $G_4$ from $F$ in equation~(\ref{magic-formula}), which is the main result of the section.

Applications of the kinematical transformation~(\ref{magic-formula}) will be treated in the following two sections. Our approach in this section, and manipulations with the Bruhat decomposition in particular, draw on ideas from~\cite{Dobrev:1977qv}.

\subsection{Weyl inversion}\label{section3.1}

When considering constraints that conformal invariance imposes on correlation functions of a~quantum field theory, an important role is played by the conformal inversion
\begin{equation*}
 I x^\mu = \frac{x^\mu}{x^2}.
\end{equation*}
This is because the conformal group is generated by translations, rotations, dilations and the inversion. Thus, often to prove some statement about all group elements, it is sufficient to show it for these four types of transformations. The first three types act linearly on spacetime and are rather simple to treat. However, the action of special conformal transformations is non-linear and it is often easier to consider~$I$ instead.

In bosonic Euclidean conformal field theory by {\it conformal group} one can mean various Lie groups that have $\mathfrak{so}(d+1,1)$ as their Lie algebra. We will assume minimal symmetry and take the bosonic conformal group $G_{\rm bos}$ to be connected and simply connected, denoted also ${\rm Spin}(d+1,1)$. Let ${\rm O}(d+1,1)$ be the group of pseudo-orthogonal matrices. We shall denote its identity component by ${\rm SO}^+(d+1,1)$. This group can be realised as the quotient of $G_{\rm bos}$ by its centre
\begin{equation*}
 {\rm SO}^+(d+1,1) = {\rm Spin}(d+1,1)/\mathbb{Z}_2 .
\end{equation*}
In particular, the inversion $I$, which is an element of ${\rm O}(d+1,1)$ not connected to the identity, is not assumed to be part of the symmetry. On the other hand, the Weyl inversion, obtained by composing $I$ with the reflection in the hyperplane orthogonal to the unit vector $e_d$, $w=s_{e_d}\circ I$, belongs to ${\rm SO}^+(d+1,1)$. It can be equivalently defined as
\begin{equation}
 w = {\rm e}^{\pi\frac{K_d-P_d}{2}} . \label{Weyl-inversion}
\end{equation}
Unlike that of ${\rm SO}^+(d+1,1)$, the action of ${\rm Spin}(d+1,1)$ on the compactified Euclidean space is not faithful, as both elements of the centre act trivially. There are two elements of ${\rm Spin}(d+1,1)$ that project to $w$ in ${\rm SO}^+(d+1,1)$. We will use the expression $(\ref{Weyl-inversion})$ as the definition of the Weyl inversion for ${\rm Spin}(d+1,1)$. Then one can check that its square is the non-trivial element of the centre, $w^2=-1$.

Now let $G$ be a superconformal group. The underlying Lie group $G_{(0)}$ of $G$ has the form
\begin{equation*}
 G_{(0)} = G_{\rm bos}\times U,
\end{equation*}
where $G_{\rm bos}$ has the Lie algebra $\mathfrak{so}(d+1,1)$ and $U$ is some Lie group describing internal symmetries of the theory. The odd part of the Lie superalgebra $\mathfrak{g} = {\rm Lie}(G)$ carries the adjoint representation of $G_{(0)}$. Furthermore, this representation decomposes into a direct sum of spinor representations under $\mathfrak{so}(d)\subset\mathfrak{g}_{(0)}$. It follows that $G_{\rm bos}$ has to be simply connected. It is for this reason that we opted to work with the simply connected group above.

For superconformal groups, we define the Weyl inversion as $w = (w_{\rm bos},e_U)$, where $e_U$ denotes the identity element of $U$. From this definition some general properties readily follow. For example, for any supertranslation generator $Q$ we have
\begin{equation*}
 \tfrac12{\rm Ad}_w(Q) = {\rm Ad}_w([D,Q]) = [{\rm Ad}_w(D),{\rm Ad}_w(Q)] = - [D,{\rm Ad}_w(Q)] .
\end{equation*}
We have used that ${\rm Ad}_{w_{\rm bos}}(D)=-D$. Therefore, the Weyl inversion interchanges generators of supertranslations and super special conformal transformations. For type I superconformal algebras we can use ${\rm Ad}_w(R) = R$ to similarly deduce
\begin{equation*}
 {\rm Ad}_w(\mathfrak{q}_\pm) = \mathfrak{s}_\pm . 
\end{equation*}
For our notation concerning the Lie algebra of the conformal group and the Lie superalgebra of the superconformal group, the reader is referred to Appendix~\ref{appendixA}.

\subsection{Bruhat decomposition}\label{section3.2}

Let us continue to denote by $G$ a superconformal group with the Lie superalgebra $\mathfrak{g} = {\rm Lie}(G)$. The subspace of elements of $\mathfrak{g}$ that have a positive dilation weight is denoted by
\begin{equation*}
 \mathfrak{m} = \mathfrak{g}_{1/2} \oplus \mathfrak{g}_1.
\end{equation*}
It is spanned by translations and supertranslations and forms a subalgebra of $\mathfrak{g}$. The corresponding subgroup $\mathcal{M}$ of $G$ is called the superspace. On the superspace, we introduce the coordinates through
\begin{equation*}
m(x) = {\rm e}^{x_{a} X^a} .
\end{equation*}
By superspace, we in fact mean a supercommutative algebra generated by elements $x^a$, which can be thought of coordinates. Thus, the algebra $\mathcal{M}$ is more appropriately thought of as the algebra of functions on the superspace rather then the space itself. The relevant notions of supergeometry and super Lie theory in particular are sketched in Appendix~\ref{appendixD}.

Let $\mathcal{M}_i$ be super-commuting copies of the superspace, where $i$ belongs to some indexing set, and let $x_i\in\mathcal{M}_i$. Given any pair of labels $i,j$ we
define the variables $x_{ij} = (x_{ija}) \in \mathcal{M}_i \otimes \mathcal{M}_j$ through
\begin{equation} \label{eq:xij}
m(x_{ij}) = m(x_j)^{-1} m(x_i) .
\end{equation}
This expression is well-defined as $\mathcal{M}$ is itself a supergroup. Concrete expressions for the components of $x_{ij}$ can be worked out from the anti-commutation relations of the supercharges $Q$.

In the last section we introduced the Weyl element $w$ through equation~\eqref{Weyl-inversion}. With the help of it, let us define a new family of supergroup elements $n$ through
\begin{equation} \label{eq:nx}
n(x) = w^{-1} m(x) w .
\end{equation}
Since $m$ involves only generators $X^a \in\mathfrak{m}$ of the superconformal algebra that raise the conformal weight, the element $n$ is built using generators $Y^a$ from the algebra $\mathfrak{n}=\lieg_{<0}$ that lower the conformal weight~-- special conformal generators $K$ and their fermionic cousins $S$.

For the bosonic conformal group, the {\it Bruhat decomposition} is a factorisation of a conformal transformation into a product of a translation, a rotation, a dilation and a special conformal transformation. It will be convenient for us to put the dilation and rotation pieces into a single factor and write $g = m(g) n(g) k(g)$.\footnote{Almost all conformal transformations admit a Bruhat decomposition. That is, the set of elements that cannot be decomposed in this way has Haar measure zero.} The corresponding decomposition of the Lie algebra reads
\begin{equation*}
 \mathfrak{g} = \mathfrak{m} \oplus \mathfrak{n} \oplus \mathfrak{k},
\end{equation*}
with $\mathfrak{k} = {\rm Lie}(K)$. The latter is a valid decomposition for any superconformal algebra. By exponentiation, it gives a decomposition of the superconformal group that we shall refer to, by a slight abuse of terminology, as the Bruhat decomposition. For an arbitrary $h\in G$ we define the functions $y(x,h)$, $z(x,h)$ and $t(x,h)$ through the factorisation\footnote{That is, if $g = h m(x)$, we define $y=y(x,h)$ and $z=z(x,h)$ by
\begin{equation*}
 m(g) = {\rm e}^{y_a X^a}, \qquad n(g) = w^{-1} {\rm e}^{z_a X^a} w,
\end{equation*}
and $t(x,h)$ by $k(g) = k(t(x,h))$ with some arbitrary coordinate system $(t^\alpha)$ on~$K$.}
\begin{equation}
 h m(x) = m(y(x,h)) n(z(x,h)) k(t(x,h)) . \label{matrix-identity}
\end{equation}
When $h=w$ we simply write $y(x)=y(x,w)$ and further $y_{ij} = y(x_{ij})$. Similarly are defined~$z(x)$, $t(x)$, $z_{ij}$ and $t_{ij}$. We have by definition
\begin{equation*} 
w m(x_{ij}) = m(y_{ij})  n(z_{ij})  k(t_{ij}).
\end{equation*}
We will regard the functions $y(x)$, $z(x)$, $t(x)$ as known. For the bosonic conformal group, they can be found in~\cite{Dobrev:1977qv}. In the case of superconformal groups, they are easily determined by manipulations with supermatrices, as shall be outlined in an example below.

\subsection{From quantum fields to functions on the group}\label{section3.3}

Fields in a superconformal theory are organised according to representations of $G$. Their transformation properties are encoded in a finite-dimensional representation of the subgroup $K = {\rm SO}(1,1)\times {\rm Spin}(d)\times U$ of dilations, rotations and internal symmetries. For bosonic theories, these labels amount to the field's conformal dimension and its spin. In the supersymmetric case, there are additional labels due to the internal symmetry group, which we collectively call $R$-charges.

Primary fields can be naturally associated with principal series representations of $G$. To make this point clear, let us focus on the bosonic theory. Principal series representations of $G={\rm SO}(d+1,1)$ can be realised on spaces of vector valued functions on the group, covariant under the (say right) regular action of the parabolic subgroup $P$ -- see Appendix~\ref{appendixA} for details. Such functions are uniquely determined by the values that they assume on~$M$, regarded as the subgroup of translations inside $G$. Under this identification, the left regular action of the Lie algebra~$\mathfrak{g}$ reads
\begin{gather*}
 p_\mu = \partial_\mu,\qquad m_{\mu\nu} = x_\nu \partial_\mu - x_\mu\partial_\nu + S_{\mu\nu},\qquad d = x^\mu\partial_\mu + \Delta_\varphi,\\ k_\mu = x^2\partial_\mu - 2 x_\mu d + 2x^\nu S_{\mu\nu} .
\end{gather*}
These are precisely the differential operators that appear in the Ward identities. In fact, we can write the Ward identities in an alternative form as follows. Let $h\in G$ be some (global) conformal transformation and $G_n$ an $n$-point correlation function of primary fields. Then
\begin{equation} \label{eq:G4Wardid}
 G_n (hx_i) = \bigg(\bigotimes_{i=1}^n \rho_i (k(t(x_i,h)))\bigg) G_n(x_i).
\end{equation}
Here $\rho_i$ are the representations of $K$ associated to fields appearing in the correlation function and $k(t(x_i,h))$ the Bruhat factors defined in the previous subsection. We will be mostly interested in four-point functions and denote the carrier space of $\rho_1\otimes\cdots \otimes\rho_4$ by $V$.

The equation~(\ref{eq:G4Wardid}) is a valid formulation of Ward identities for conformal and superconformal theories alike. Having written the identities in this way, we can state and prove our first result:

\begin{Theorem} There is a $1$-$1$ correspondence between solutions of Ward identities $(\ref{eq:G4Wardid})$ for four-point functions and that of $K$-spherical functions on the superconformal group~$G$. The latter are elements of the algebra $A(G)\otimes V$ which satisfy, for all $k_l,k_r\in K$
\begin{gather}\label{eq:covariance}
 F(k_l g k_r)= \big(\rho_1(k_l)\otimes\rho_2\big(w k_l w^{-1}\big)\otimes\rho_3\big(k_r^{-1}\big)\otimes\rho_4\big(w k_r^{-1}w^{-1}\big)\big) F(g).
\end{gather}
\end{Theorem}

\begin{proof} We would first like to show how a solution of Ward identities can be used to produce a~$K$-spherical function. First, any solution $G_4$ can be extended in a unique way to a function~$F_4$ on four copies of~$G$ if we impose
\begin{equation}\label{eq:F4rightcov}
 F_4(m(x_i)) = G_4(x_i) ,\qquad F_4 (g_i n_i k_i) = \bigotimes_{i=1}^{4} \rho_i \big(k_i^{-1}\big)
 F_4(g_i).
\end{equation}
The Ward identities \eqref{eq:G4Wardid} satisfied by $G_4$ imply the following invariance conditions satisfied by~$F_4$ under the diagonal left regular action of $G$
\begin{align}
 F_4(h g_i) &= F_4(h m(x_i) n_i k_i) = F_4\big( m\big(x_i^h\big) n(z(x_i,h)) k(t(x_i,h)) n_i k_i\big) \nonumber\\
 &= \bigg( \bigotimes_{i=1}^4 \rho_i(k_i)^{-1} \bigg)\bigg( \bigotimes_{i=1}^4 \rho_i \big(k(t(x_i,h))^{-1}\big) \bigg) G_4\big(x_i^h\big) \nonumber\\
 & = \bigg( \bigotimes_{i=1}^4 \rho_i(k_i)^{-1} \bigg)G_4(x_i) = F_4(g_i) .\label{F4-left-invariance}
\end{align}
In this short derivation, besides Ward identities, we used the definition~(\ref{matrix-identity}) and covariance laws~(\ref{eq:F4rightcov}). Given $F_4$ and the Weyl inversion $w$ we can construct a new object $F\in\ A(G)\otimes V$ by
\begin{equation}\label{eq:FfromF4}
 F(g) := F_4 \big(e,w^{-1},g,gw^{-1}\big).
\end{equation}
While the motivation for such a map might not be clear, it is readily verified that $F$ is a $K$-spherical function. Indeed, from the definition \eqref{eq:FfromF4} of $F$, the left invariance condition~(\ref{F4-left-invariance}) and the right covariance law in equation~\eqref{eq:F4rightcov} of $F_4$ we obtain
\begin{align*}
 F(k_l g k_r) &= F_4\big(e,w^{-1},k_l g k_r, k_l g k_r w^{-1}\big) =
 F_4\big(k_l^{-1},w^{-1} w k_l^{-1}w^{-1},g k_r, g w^{-1} w k_r w^{-1} \big)\\
 &= \big(\rho_1(k_l)\otimes\rho_2\big(w k_l w^{-1}\big)
 \otimes\rho_3\big(k_r^{-1}\big)\otimes\rho_4\big(wk_r^{-1}w^{-1}\big)\big) F(g).
\end{align*}
We shall now go in the other direction and show how to recover~$G_4$ from~$F$. Suppressing the last two arguments and their corresponding prefactors for simplicity, we have
\begin{align*}
F_4(m(x_1),m(x_2)) & = \big(1 \otimes \rho_2\big(k(t_{21})^{-1}\big)\big)
F_4\big(m(x_1) n(y_{21}), m(x_2) k(t_{21})^{-1} n(z_{21})^{-1}\big) \\
& = \big(1 \otimes \rho_2\big(k(t_{21})^{-1}\big)\big)
F_4\big(m(x_1) n(y_{21}), m(x_1) m(x_{21}) k(t_{21})^{-1} n(z_{21})^{-1}\big) \\
& = \big(1 \otimes \rho_2\big(k(t_{21})^{-1}\big)\big)
F_4\big(m(x_1) n(y_{21}), m(x_1) w^{-1} m(y_{21})\big) \\
& = \big(1 \otimes \rho_2\big(k(t_{21})^{-1}\big)\big)
F_4\big(m(x_1) n(y_{21}), m(x_1) n(y_{21}) w^{-1} \big).
\end{align*}
In the first step we used the right covariance property~\eqref{eq:F4rightcov} of $F_4$ and the fact that the compensating prefactors are trivial on elements of the form $n(x)$. Next, we inserted $m(x_{21})$ using its definition~(\ref{eq:xij}) and applied the formula
\[ m(x_{21}) = w^{-1} m(y_{21}) n(z_{21}) k(t_{21}),\]
which is essentially the definition of $y_{21}$, $z_{21}$ and~$t_{21}$. Finally we commuted the element~$w^{-1}$ past~$m(y_{21})$ by an application of~(\ref{eq:nx}). The same steps can be repeated for the second two arguments to arrive at
\begin{gather*}
F_4(m(x_i)) = \big(1 \otimes \rho_2\big(k(t_{21})^{-1}\big) \otimes 1 \otimes \rho_4\big(k(t_{43})^{-1}\big)\big)\\
\hphantom{F_4(m(x_i)) =}{}\times F_4\big(g_{12}(x_i),g_{12}(x_i) w^{-1}, g_{34}(x_i), g_{34}(x_i) w^{-1}\big),
\end{gather*}
where we introduced the elements
\[ g_{ij}= m(x_i) n(y_{ji}) .\]
To complete the derivation, we use the left invariance property of $F$ proved in~(\ref{F4-left-invariance}), with $h = g_{12}^{-1}$
\[ F_4(m(x_i)) = \big(1 \otimes \rho_2\big(k(t_{21})^{-1}\big)\otimes 1 \otimes \rho_4\big(k(t_{43})^{-1}\big)\big)F_4\big(e, w^{-1}, g(x_i), g(x_i) w^{-1}\big),\]
where the element $g(x_i)$ is defined as
\begin{equation*}
 g(x_i) = g_{12}^{-1} g_{34} = n(y_{21})^{-1} m(x_{31}) n(y_{43}) .
\end{equation*}
Putting everything together, the correlation function $G_4$ is recovered from the corresponding $K$-spherical function $F$ as
\begin{equation}\label{magic-formula}
G_4(x_i) = \big(1\otimes\rho_2(k(t_{21}))^{-1}\otimes1\otimes\rho_4(k(t_{43}))^{-1}\big) F(g(x_i)) .
\end{equation}
This establishes the theorem. The last relation will be referred to as the {\it lifting formula} and it is the main result of this section.
\end{proof}

Let us make a few remarks on the uniqueness of such a lift. As can be seen, the formula consists of two ingredients, the argument $g(x_i)$ of the function on the right-hand side and the prefactor. Also, the space of $K$-spherical functions depends on the quantum numbers of the fields in the correlation function. These different ingredients play somewhat different roles in the lift of the correlator. To explain these, assume for the moment that all fields transform trivially under rotations, dilations and internal symmetries. In this case, the space~(\ref{eq:covariance}) is that of $K$-$K$ {\it invariant} functions and the prefactor in~(\ref{magic-formula}) is trivial. It is the function $g(x_i)$, which does not depend on quantum numbers, that ensures that differential equations of Ward identities are carried to $K$-$K$ invariance laws. In this sense, the group element $g(x_i)$ is the most fundamental part of~(\ref{magic-formula}). When we allow for fields with non-trivial quantum numbers, the $K$-$K$ invariant functions should be modified to the space $(\ref{eq:covariance})$. These new covariance laws are dictated by lifts of individual fields. Once the element $g(x_i)$ and the space~(\ref{eq:covariance}) are fixed, the prefactor in~(\ref{magic-formula}) is the unique one which ensures that $G_4(x_i)$ satisfies the Ward identities.

\subsection{Example}\label{section3.4}

Throughout this work, we will illustrate the general constructions on the example of $\mathcal{N}=2$ superconformal symmetry in $d=1$ dimension. The complexified superconformal algebra for this case is $\mathfrak{g} = \mathfrak{sl}(2|1)$. Its bosonic subalgebra $\mathfrak{g}_{(0)}$ is spanned by the generators of dilations~$D$, translations~$P$, special conformal transformations $K$, and internal symmetry $R$. The odd subspace is four-dimensional and spanned by supertranslations~$Q_\pm$ and super special conformal transformations~$S_\pm$.

Many of our computations are most easily performed by working with a concrete representation of $\mathfrak{g}$. The smallest faithful representation of $\mathfrak{g}$ is 3-dimensional. We may choose the generators as
\begin{gather*} 
 D = \begin{pmatrix}
 1/2 & 0 & 0\\
 0 & -1/2 & 0\\
 0 & 0 & 0
 \end{pmatrix},\qquad P = \begin{pmatrix}
 0 & 1 & 0\\
 0 & 0 & 0\\
 0 & 0 & 0
 \end{pmatrix},\\
 K = \begin{pmatrix}
 0 & 0 & 0\\
 1 & 0 & 0\\
 0 & 0 & 0
 \end{pmatrix},\qquad R = \begin{pmatrix}
 -1 & 0 & 0\\
 0 & -1 & 0\\
 0 & 0 & -2
 \end{pmatrix},
\end{gather*}
for the four bosonic generators and
\begin{gather*} 
 Q_- = \begin{pmatrix}
 0 & 0 & 0\\
 0 & 0 & 0\\
 0 & 1 & 0
 \end{pmatrix},\qquad\! Q_+ = \begin{pmatrix}
 0 & 0 & 1\\
 0 & 0 & 0\\
 0 & 0 & 0
 \end{pmatrix},\qquad\! S_- = \begin{pmatrix}
 0 & 0 & 0\\
 0 & 0 & 0\\
 1 & 0 & 0
 \end{pmatrix},\qquad\! S_+ = \begin{pmatrix}
 0 & 0 & 0\\
 0 & 0 & 1\\
 0 & 0 & 0
 \end{pmatrix},
\end{gather*}
for the fermionic ones.

The superspace corresponding to $\mathfrak{g}$ is $\mathcal{M} = \mathbb{R}^{1|2}$. Its structure algebra is generated by one bosonic variable $u$ along with two Grassmann variables~$\theta$ and $\bar \theta$, which we write collectively as $x=(u,\theta,\bar\theta)$. The supergroup element $m$ we introduced above takes the following matrix form
\begin{equation*} 
m(x) = {\rm e}^{u P + \theta Q_+ + \bar \theta Q_-} = \begin{pmatrix}
 1 & X & \theta \\
 0 & 1 & 0 \\
 0 & -\bar\theta & 1
 \end{pmatrix} ,
\end{equation*}
where $X = u-\frac12 \theta \bar \theta$. The supergroup structure of $\mathcal{M}$ enabled us to define variables $x_{ij}$ in $(\ref{eq:xij})$. These can be determined by matrix multiplication
\begin{equation*}
 u_{ij} = u_i - u_j -\tfrac12\theta_i\bar\theta_j - \tfrac12\bar\theta_i\theta_j  ,\qquad
 \theta_{ij} = \theta_i - \theta_j  ,\qquad \bar\theta_{ij} = \bar\theta_i - \bar\theta_j.
\end{equation*}
Next, we turn to the Weyl inversion and supergroup elements built out of special superconformal transformations that were denoted $n(x)$. According to~(\ref{Weyl-inversion}) and~(\ref{eq:nx})
\begin{equation*} 
 w = {\rm e}^{\pi\frac{K-P}{2}} = \begin{pmatrix}
 0 & -1 & 0\\
 1 & 0 & 0\\
 0 & 0 & 1
 \end{pmatrix},\qquad n(x) = w^{-1} m(x) w = \begin{pmatrix}
 1 & 0 & 0\\
 -X & 1 & -\theta\\
 -\bar\theta & 0 & 1
 \end{pmatrix} .
\end{equation*}
These ingredients suffice to determine the factors of the Bruhat decomposition~(\ref{matrix-identity}) with $h=w$. It is expressed by the matrix identity
\begin{gather*}
 \begin{pmatrix}
 0 & -1 & 0\\
 1 & X & \theta\\
 0 & -\bar\theta & 1
 \end{pmatrix} = \begin{pmatrix}
 1 & -\frac1u \big(1+\frac{\theta\bar\theta}{2u}\big) & \theta/u\\
 0 & 1 & 0\\
 0 & -\bar\theta/u & 1
 \end{pmatrix} \begin{pmatrix}
 1 & 0 & 0\\
 u+\frac12\theta\bar\theta & 1 & \theta\vspace{1mm}\\
 \bar\theta & 0 & 1
 \end{pmatrix} \\
 \hphantom{\begin{pmatrix}
 0 & -1 & 0\\
 1 & X & \theta\\
 0 & -\bar\theta & 1
 \end{pmatrix} =}{}\times \begin{pmatrix}
 \frac1u\big(1-\frac{\theta\bar\theta}{2u}\big) & 0 & 0\\
 0 & u\big(1-\frac{\theta\bar\theta}{2u}\big) & 0\\
 0 & 0 & 1-\frac{\theta\bar\theta}{u}
 \end{pmatrix}  , 
\end{gather*}
from which one reads off the various factors
\begin{equation*}
 y(x) = w(x) = \left(\frac{-1}{u},\frac{\theta}{u},\frac{\bar\theta}{u}\right),\qquad z(x) = (-u,-\theta,-\bar\theta), \qquad k(t(x)) = {\rm e}^{-\log u^2 D + \frac{\theta\bar\theta}{2u}R} .
\end{equation*}
Let $G_4(x_i)$ be a four-point function of superconformal primary fields. These are labelled by conformal weights $\Delta_i$ and $R$-charges $r_i$. Our conventions are such that the corresponding representation of $K={\rm SO}(1,1)\times {\rm SO}(2)$ reads
\begin{equation*} \label{eq:rho-1d}
 \rho_{\Delta,r}\big({\rm e}^{\lambda D + \kappa R}\big) = {\rm e}^{-\Delta\lambda + r\kappa} .
\end{equation*}
The formula~(\ref{magic-formula}) now states that the correlator $G_4$ admits a unique representation
\begin{equation}\label{magic-1d}
 G_4(x_i) = \frac{{\rm e}^{r_2\frac{\theta_{12}\bar\theta_{12}}{2u_{12}}+r_4\frac{\theta_{34}\bar\theta_{34}}{2u_{34}}}}{u_{12}^{2\Delta_2}u_{34}^{2\Delta_4}} F\big({\rm e}^{-w(x_{21})\cdot X^w} {\rm e}^{x_{31}\cdot X}{\rm e}^{w(x_{43})\cdot X^w}\big),
\end{equation}
where $F$ is a $K$-spherical function on $G={\rm SL}(2|1)$, i.e., one that obeys
\begin{gather*}
 F\big({\rm e}^{\lambda_l D + \kappa_l R} g {\rm e}^{\lambda_r D + \kappa_r R}\big)={\rm e}^{(\Delta_2-\Delta_1)\lambda_l+(r_1+r_2)\kappa_l} {\rm e}^{(\Delta_3-\Delta_4)\lambda_r- (r_3+r_4)\kappa_r} F(g).
\end{gather*}
In~(\ref{magic-1d}) we used the notation $X = (P,Q_+,Q_-)$ and $X^w = w^{-1}(P,Q_+,Q_-)w = (-K,-S_+,S_-)$. With this equation, we conclude the discussion of the $\mathfrak{sl}(2|1)$ example for the present. The eigenbasis of the Laplacian in the space of $K$-spherical functions on ${\rm SL}(2|1)$ will be studied in Section~\ref{section5}.

\section{Superconformal crossing equations}\label{section4}

In this section we will construct the {\it crossing factor}, a matrix that is roughly defined as the ratio of tensor factors in two different channels. Namely, if the fields in the four-point function transform non-trivially under rotations, the Ward identities constrain the correlator to take the form
\begin{equation}\label{tensor-structures}
 G_4^\alpha(x_i) = \Omega(x_i)^\alpha_{\ I} F^I(u_a) .
\end{equation}
Here, $\alpha$ runs over a basis for the space of polarisations $V$ of the four fields and $I$ runs over the space for the space of four-point tensor structures $W$. If $d>3$ we have in general $\dim W<\dim V$. Detailed discussions of tensor structures can be found in \cite{Costa:2011mg,Cuomo:2017wme,Kravchuk:2016qvl,Osborn:1993cr, Schomerus:2016epl}. We will summarise some of the main points that will be important in the later subsections. For concreteness, let us focus on bosonic theories.

The configuration space $M^4$ of four points is foliated into orbits of $G$ under the diagonal action and the four-point function is completely specified by giving its values on one point of each orbit. Let us denote the space of orbits by $X=M^4/G$. The structure of this space might be complicated, but there is an open dense subset of $X$ which is a smooth manifold with local coordinates $(u_a)$. Since the action of $G$ on $M^4$ is not free, not every function $X\xrightarrow{}V$ gives a well-defined correlation function. To see this, let $x_1,\dots,x_4$ be four points in general position. The stabiliser of $(x_1,\dots,x_4)$ in $G$ under the diagonal action is isomorphic to ${\rm Spin}(d-2)$. Indeed, one can notice that this is the stabiliser when points are chosen as in~(\ref{points-choice}) -- it consists of rotations of the space spanned by vectors $e_3,\dots,e_d$. For other choices of the four points, the stabiliser subgroup is related to this one by conjugation. Let us denote the points from~(\ref{points-choice}) by~$x_i^0$ and their stabiliser by $B_{\rm bos}$. For any~$b\in B_{\rm bos}$, the Ward identities imply
\begin{equation*}
 G_4\big(x_i^0\big) = \big(\rho_1\big(db_{x_1^0}\big)\otimes\cdots \otimes\rho_4\big(db_{x_4^0}\big)\big) G_4\big(x_i^0\big) = (\rho_1(b)\otimes\cdots \otimes\rho_4(b) ) G_4\big(x_i^0\big),
\end{equation*}
where in the last equality we used that all elements of $B_{\rm bos}$ act on $M$ as linear transformations. In conclusion, $G_4\big(x_i^0\big)$ belongs to the space of invariants $V^{B_{\rm bos}}$. As a vector space, this is the direct sum of trivial representations of $B_{\rm bos}$ that appear in the decomposition
\begin{equation*}
{\rm Res}^K_{B_{\rm bos}}(\rho_1\otimes\cdots \otimes\rho_4) .
\end{equation*}
A generic orbit in $M^4$ contains a point of the form $\chi=(0,\infty,x_3,x_4)$ and corresponding spaces $V^{{\rm Stab}_G(\chi)}$ all have the same dimension. This allows to write the correlation function as in~(\ref{tensor-structures}), where $\dim W =\dim V^{B_{\rm bos}}$.

When the equation $(\ref{tensor-structures})$ is written for two different permutations of the points $x_i$ and the permutation symmetry of Euclidean correlation functions is used, one arrives at the following generalisation of $(\ref{crossing-eqn-scalars})$
\begin{equation*}
 F^I(u_a) = \mathcal{M}^I_{\ J}(u_a) F^J(u'_a) .
\end{equation*}
The matrix $\mathcal{M}^I_{\ J}$ is termed the crossing factor. A distinguishing feature of the crossing factor is its (super)conformal invariance. Therefore, whereas the tensor structures depend non-trivially on coordinates of all insertion points, the crossing factor is a function of cross ratios only. This fact can be used to compute the factor in a simple manner. Let us note that to do conformal bootstrap, the knowledge of $\mathcal{M}^I_{\ J}$ in general has to be supplemented by the analysis of 3-point tensor structures.

In the first subsection we will prove a useful proposition about the transformation properties of various group elements introduced in the previous section under superconformal transformations. In the second subsection, we define the crossing factor and prove its superconformal invariance. The third subsection contains the computation of the factor for general spinning fields in bosonic field theories, while the fourth treats the example of $\mathfrak{sl}(2|1)$ superconformal symmetry.

\subsection{Transformations of Bruhat factors}\label{section4.1}

In the previous section, we considered the Bruhat decomposition~(\ref{matrix-identity}) and its specialisation to $h=w$. We shall now study the ``functorial properties'' of these factors when $x$ is acted on by a superconformal transformation $h$. In particular, the following transformation laws can be established

\begin{Proposition} Under a superconformal transformation $h$, elements $g_{ij}$
and $k(t_{ji})$ transform as
\begin{equation}\label{proposition-tranformation-laws}
 g_{ij} \big(x^h\big) = h g_{ij}(x) k(t(x_i,h))^{-1},\qquad k\big(t_{ji}^h\big) = w k(t(x_i,h)) w^{-1} k(t_{ji}) k(t(x_j,h))^{-1}.
\end{equation}
Let us make a comment on the notation. Various objects in this subsection depend on the insertion points~$x_i$. However, to avoid having long expressions we have not explicitly kept this dependence in the notation, e.g., we write $y_{ij} = y(x_{ij})$ etc. We will adopt the rule that if the insertion points are transformed by a group element $h$, the corresponding objects will carry an upper index~$h$, e.g., $y_{ij}^h$, $t_{ij}^h$ etc. In particular, we alternatively write $x^h$ or $hx$ $($or very rarely~$y(x,h))$.
\end{Proposition}

\begin{proof} Consider the system of equations
\begin{align}\label{system-1}
 & m(x_i) n(y_{ji}) = g_{ij}(x), \qquad m(x_j) k(t_{ji})^{-1} n(z_{ji})^{-1} = g_{ij}(x)w^{-1} .
\end{align}
The first equation is the definition of $g_{ij}(x)$ and the second one was proved in the previous section. Let us apply a transformation $h$ to all $x_i$-s and use
\begin{equation*}
 m(x^h) = h m(x) k(t(x,h))^{-1} n(z(x,h))^{-1} .
\end{equation*}
This relation follows at once from definitions of $k(t(x,h))$ and $n(z(x,h))$. Doing these two steps, we get another system of equations
\begin{gather}\label{system-2}
 h m(x_i) k(t(x_i,h))^{-1} n(z(x_i,h))^{-1} n\big(y_{ji}^h\big) = g_{ij}\big(x^h\big),\\
h m(x_j) k(t(x_j,h))^{-1} n(z(x_j,h))^{-1} k\big(t^h_{ji}\big)^{-1} n\big(z^h_{ji}\big)^{-1} = g_{ij}\big(x^h\big)w^{-1}.\label{system-2-2}
\end{gather}
We can compare this system to~(\ref{system-1}). Elements $g_{ij}(x)$ and $h^{-1}g_{ij}\big(x^h\big)$ have the same $m$ Bruhat factor and similarly $g_{ij}(x)w^{-1}$ and $h^{-1}g_{ij}\big(x^h\big)w^{-1}$. It follows that they are related by
\begin{align}\label{main-step}
 h^{-1}g_{ij}\big(x^h\big) = g_{ij}(x) k_{ij} n_{ij},\qquad h^{-1} g_{ij}\big(x^h\big) w^{-1} = g_{ij}(x) w^{-1} k'_{ij} n'_{ij},
\end{align}
for some $k_{ij}$, $k'_{ij}$, $n_{ij}$, $n'_{ij}$. Putting these two equations together, we have
\begin{gather*}
 k_{ij} n_{ij} = \big(w^{-1} k'_{ij}w\big) (w^{-1} n'_{ij} w) .
\end{gather*}
We now make the key observation -- the grading with respect to the dilation weight requires $n_{ij} = n'_{ij}=1$. Also, by looking at elements of conformal weight zero in the first equation of~(\ref{system-1}) and~(\ref{system-2}). we see that $k_{ij} = k(t(x_i,h))^{-1}$. Having established these facts, the proposition follows from~(\ref{main-step}). To get the first claim, one simply substitutes the expressions for $k_{ij}$ and $n_{ij}$ into the first equation. The second claim requires a few more steps. Let us begin by substituting $n'_{ij}=1$ and $k'_{ij} = w k_{ij} w^{-1}$ into the second equation in~(\ref{main-step}). After cancelling $w^{-1}$ factors on the right
\begin{equation*}
 h^{-1} g_{ij}\big(x^h\big) = g_{ij}(x)k(t(x_i,h))^{-1} .
\end{equation*}
Next, we use~(\ref{system-2-2}) and the second equation of~(\ref{system-1}) the to expand $g_{ij}(x^h)$ and $g_{ij}(x)$ on the two sides and cancel the $m(x_j)$ factors
\begin{equation*}
 k(t(x_j,h))^{-1} n(z(x_j,h))^{-1} k\big(t^h_{ji}\big)^{-1} n\big(z^h_{ji}\big)^{-1} w = k(t_{ji})^{-1} n(z_{ji})^{-1} w k(t(x_i,h))^{-1} .
\end{equation*}
The grading on $\mathfrak{g}$ allows to equate the $k$-factors from the two sides
\begin{equation*}
 k(t(x_j,h))^{-1} k\big(t^h_{ji}\big)^{-1} = k(t_{ji})^{-1} w k(t(x_i,h))^{-1} w^{-1} .
\end{equation*}
Rearranging terms now gives the second claim and completes the proof of the proposition.
\end{proof}

\subsection{Cartan coordinates and the crossing factor}\label{section4.2}

By this point, the usefulness of the result (\ref{magic-formula}) may not be clear. This formula says that conformal four-point functions may be regarded as $K$-spherical functions. The benefit of such transformation is that the latter can be analysed with established techniques of group theory. The first such technique is the Cartan decomposition that we shall now introduce.

Consider for the moment the bosonic conformal group $G = {\rm Spin}(d+1,1)$. Almost all elements $g\in G$ can be written in the form
\begin{equation*}
 g = k_l a k_r,
\end{equation*}
where $k_l, k_r\in K={\rm SO}(1,1)\times {\rm Spin}(d)$ and $A$ is the two-dimensional abelian group generated by $\{P_1+K_1,P_2-K_2\}$. We shall parametrise $A$ by local coordinates $(u_1,u_2)$ according to
\begin{equation*}
a(u_1,u_2) = {\rm e}^{\frac{u_1+u_2}{4}(P_1+K_1) -{\rm i} \frac{u_1-u_2}{4}(P_2 - K_2)} .
\end{equation*}
A vector-valued function on $G$ that is covariant with respect to both left and right regular actions of $K$ is uniquely specified by the values it takes on~$A$. Therefore, the $K$-spherical functions may be considered as functions of two variables $u_1$, $u_2$, \cite{Schomerus:2017eny, Schomerus:2016epl}. This is in agreement with the fact that four-point functions in a conformal field theory depend on two cross ratios.

It is well-known that in superconformal theories, one has additional fermionic invariants on which four-point functions depend. This can also be understood from a generalisation of the Cartan decomposition to the supersymmetric setup. The generalisation is achieved as follows. Supergroup elements are written as
\begin{equation} \label{eq:sCartan}
g = k_l \eta_l a \eta_r k_r ,
\end{equation}
where $k_{l}$ and $k_{r}$ are associated with the subgroup $K$ of rotations, dilations and R-symmetry transformations. The factors $\eta_l$ and $\eta_r$ are associated with fermionic generators. More specifically, $\eta_l$ is obtained by exponentiation of generators of negative $R$-charge and $\eta_r$ from generators with positive charge. From now on, we assume that the superconformal algebra $\mathfrak{g}$ is of type~I~-- see appendix for the definition and the list of such algebras. For the time being, we only note that in any superconformal algebra of type I there exists a distinguished $\mathfrak{u}(1)$ subalgebra defined by~(\ref{U1R}), which is a part of internal symmetry algebra~$U_r$. The generator of this subalgebra will be denoted by~$R$. Half of the fermionic generators half positive $R$-charge, while the others have negative. One way to arrive at the decomposition~(\ref{eq:sCartan}) will be explained below,~(\ref{param}).

The factorisation of supergroup elements $g$ in the form \eqref{eq:sCartan} is not unique. In fact, given any such factorisation we can produce another one by the transformation
\begin{equation} \label{eq:gaugeB}
 (k_l,\eta_l;k_r,\eta_r ) \rightarrow \big(k_l b, b^{-1} \eta_l b ; b^{-1} k_r, b^{-1} \eta_r b\big),
\end{equation}
where $b\in B = {\rm Spin}(d-2) \times U_r$. This is how the stabiliser group of four points appears in the Calogero--Sutherland coordinates. The elements of~$B$ commutes with $a = a(u_1,u_2)$. At the same time, the elements $b^{-1} \eta_{l,r} b$ can still be written as exponentials of fermionic generators with negative~$(l)$ and positive~$(r)$ $U(1)$ $R$-charge, respectively. Hence the gauge transformation~\eqref{eq:gaugeB} respects the Cartan decomposition. In the following, we shall fix the Cartan factors of group elements in some arbitrary way, and refer to this choice as gauge fixing. It will be shown that all quantities that are of interest do not depend on this choice.

Let us now return to the equation~(\ref{magic-formula}). This equation treats each of the four insertion points differently and hence it breaks the permutation symmetry of correlators in a Euclidean quantum field theory. We will be concerned two particular permutations, $\sigma_s = 1$ and $\sigma_t = (2 4)$, customarily called the $s$-channel and the $t$-channel. Given any choice of $\sigma$, we can extend the lifting formula~\eqref{magic-formula} to become
\begin{equation} \label{eq:magic-formulasigma}
G_4(x_i) = \rho_{\sigma(2)}\big(k(t_{\sigma(2)\sigma(1)})^{-1}\big) \rho_{\sigma(4)}\big(k(t_{\sigma(4)\sigma(3)})^{-1}\big) F_\sigma(g(x_{\sigma(i)})).
\end{equation}
This equation defines $F_\sigma$. The factor $\rho_{\sigma(i)}$ acts on the $\sigma(i)^{\rm th}$ tensor factor in $V = V_1\otimes \cdots \otimes V_4$ and it acts trivially on all other tensor factors. To evaluate~(\ref{eq:magic-formulasigma}) further, we
decompose the argument $g(x_{\sigma(i)})$ of the functional $F_\sigma$ in Cartan factors
\begin{equation*}
 g(x_{\sigma(i)}) = k_{\sigma,l}(x_i) \eta_{\sigma,l}(x_i)a_\sigma(x_i)\eta_{\sigma,r}(x_i) k_{\sigma,r}(x_i).
\end{equation*}
The formula \eqref{eq:magic-formulasigma} and covariance properties of $F_\sigma$ give
\begin{align*}
G_4(x_i) & = \rho_{\sigma(2)}\big(k(t_{\sigma(2)\sigma(1)})^{-1}\big)
\rho_{\sigma(4)}\big(k(t_{\sigma(4)\sigma(3)})^{-1}\big) F_\sigma(g(x_{\sigma(i)})) \\
 & = \rho_{\sigma(2)} (k(t_{\sigma(2)\sigma(1)}) )^{-1} \rho_{\sigma(4)}
 (k(t_{\sigma(4)\sigma(3)}) )^{-1} F_\sigma(k_{\sigma,l}\eta_{\sigma,l}
a_\sigma \eta_{\sigma,r}k_{\sigma,r}) \\
&  = \rho_{\sigma(1)}(k_{\sigma,l})\rho_{\sigma(2)}\big(k(t_{\sigma(2)\sigma(1)})^{-1}
k_{\sigma,l}^{w}\big)\rho_{\sigma(3)}\big(k_{\sigma,r}^{-1}\big)\rho_{\sigma(4)}
\big(k(t_{\sigma(4)\sigma(3)})^{-1} \big(k^{-1}_{\sigma,r}\big)^{w}\big)\\
&\quad{}\times
F_\sigma(\eta_{\sigma,l}a_\sigma \eta_{\sigma,r}) .
\end{align*}
For simplicity of notation, we omitted the dependence of Cartan factors on the insertion points and wrote $k_{\sigma,l} = k_{\sigma,l}(x_i) = k_l(x_{\sigma(i)})$ etc. In the $s$- and $t$-channels, the last formula reduces to
\begin{gather*} 
G_4(x_i) = \rho_1 (k_{s,l}) \rho_2\big(k(t_{21})^{-1}k^{w}_{s,l}\big) \rho_3\big(k_{s,r}^{-1}\big)\rho_{4}\big(k(t_{43})^{-1} \big(k^{w}_{s,r}\big)^{-1}\big) F_s(\eta_{s,l}a_s \eta_{s,r}), \\
G_4(x_i) = \rho_1 (k_{t,l}) \rho_4\big(k(t_{41})^{-1} k^{w}_{t,l}\big) \rho_3\big(k_{t,r}^{-1}\big)\rho_{2}\big(k(t_{23})^{-1} \big(k^{w}_{t,r}\big)^{-1}\big) F_t(\eta_{t,l}a_t \eta_{t,r}) . 
\end{gather*}
The factors in front of $F_s$ and $F_t$ are possible choices for the prefactor $\Omega(x_i)^\alpha_{\ I}$ in $s$- and $t$-channels (this factor is not unique, as it can always be multiplied by a function of cross ratios). They carry the dependence of $G_4(x_i)$ that is fixed by Ward identities. It follows that the crossing factor $\mathcal{M} = \mathcal{M}_{\rm st}$ can be obtained from the matrix $\mathcal{C}^\alpha_{\ \beta} = \rho_1(\kappa_1)\otimes\cdots \otimes\rho_4(\kappa_4)$ with
 \begin{alignat*}{3}
& \kappa_1 = k_{t,l}^{-1}k_{s,l}  ,  \qquad &&
\kappa_2 = k^w_{t,r} k(t_{23})k(t_{21})^{-1} k^w_{s,l}, &\\
& \kappa_3 = k_{t,r}k_{s,r}^{-1}  ,  \qquad &&
\kappa_4 = \big(k^w_{t,l}\big)^{-1} k(t_{41})k(t_{43})^{-1} \big(k^w_{s,r}\big)^{-1}.&
\end{alignat*}
As indicated by its indices, the matrix $\mathcal{C}^\alpha_{\ \beta}$ acts on the space of all polarisations, rather than on the space of tensor structures. To get to the crossing factor, one has to project to the space of tensor structures and this is done by restricting $\mathcal{C}$ to $B$-invariants, $\mathcal{M}^I_{\ J} = (\mathcal{C}^\alpha_{\ \beta})^B$. In order to compute the matrix $\mathcal{C}$ we first show that it is invariant under superconformal transformation, up to gauge transformations. This then implies that the projection~$\mathcal{M}$ is a function of cross ratios only and so it can be computed after moving the insertion points into a special positions.

To show invariance of $\mathcal{M}$ we study the dependence of each of the factor $\kappa_i$ on the insertion points. In this endeavour, the key role is played by the proposition $(\ref{proposition-tranformation-laws})$ proved in the first subsection. As its immediate consequence, one observes that the supergroup elements $g_\sigma(x_i)$ transform as
\begin{equation*}
g_\sigma\big(x_i^h\big) = k(t(x_{\sigma(1)},h)) g_\sigma(x_i) k(t(x_{\sigma(3)},h))^{-1} .
\end{equation*}
Because of the gauge freedom of the Cartan decomposition which we described in~\eqref{eq:gaugeB}, knowing the behaviour of $g_\sigma(x_i)$ under conformal transformations does not allow us to uniquely
determine the transformation law of the factors, but we can conclude that
\begin{gather*}
k_{\sigma,l}\big(x^h_i\big) = k(t(x_{\sigma(1)},h)) k_{\sigma,l}(x_i) b_\sigma(x_i,h) , \qquad
k_{\sigma,r}\big(x^h_i\big) = b^{-1}_\sigma(x_i,h) k_{\sigma,r}(x_i) k(t(x_{\sigma(3)},h))^{-1},
\end{gather*}
for some factor $b$ that may depend on the channel, the superspace insertion points~$x_i$ and the superconformal transformation~$h$, yet must be the same for the left and right
factors $k_l$ and $k_r$. For the case of $s$- and $t$-channels, these become
\begin{equation*}
k_{s/t,l}\big(x^h_i\big) = k(t(x_{1},h)) k_{s/t,l} b_{s/t}(x_i,h) , \qquad
k_{s/t,r}\big(x^h_i\big) = b_{s/t}^{-1}(x_i,h) k_{s/t,r} k(t(x_{3},h))^{-1}.
\end{equation*}
With these transformation laws it is now easy to verify that all four tensor components $\kappa_i$ of~$\mathcal{M}_{\rm st}$ are indeed invariant under superconformal transformations, up to gauge transformations
\begin{equation*}
\kappa_i\big(x^h_k\big) = b^{-1}_t(x_k,h) \kappa_i(x_k) b_s(x_k,h)  , \qquad
\kappa_j\big(x^h_k\big) = w b^{-1}_t(x_kh) w^{-1} \kappa_j(x_k) w b_s(x_k,h) w^{-1} ,
\end{equation*}
where $i=1,3$ and $j=2,4$. To get the last two relations one employs the formula for~$\big(t^h_{ji}\big)$ given in~(\ref{proposition-tranformation-laws}).

We have shown that the matrix $\mathcal{C}$ is invariant up to a gauge transformation. This actually implies superconformal invariance of the super-crossing factor $\mathcal{M}$. Details concerning gauge-independence can be found in~\cite{Buric:2020buk}.

\subsection{Crossing factors in bosonic theories}\label{section4.3}

The analysis we have performed in the previous subsections holds for conformal and superconformal symmetries alike. We shall now evaluate the crossing factor $\mathcal{M}_{\rm st}$ for spinning correlators in bosonic conformal field theories. In this case, the problem actually reduces to one on the 2-dimensional conformal group, as we shall show presently.

Let us deviate from previous notations and use $G$ to denote the bosonic conformal group ${\rm SO}(d+1,1)$ and assume $d>2$. Since the crossing factor is conformally invariant, in computing~$\mathcal{M}(u,v)$ we may assume that~$x_i$ are any points that give the correct cross ratios $u$ and $v$. In particular, all points can be assumed to lie in the 2-dimensional spanned by the first two unit vectors $e_1$, $e_2$ of the $d$-dimensional space $\mathbb{R}^d$. In this case, the elements $m(x_i)$, $n(x_i)$, $g_{ij}$ and~$g(x_i)$ all belong to the conformal group of the plane $G_P={\rm SO}(3,1)\subset G$. Within $G_P$, the element $g_\sigma(x_i)$ admits a unique Cartan decomposition. However, since the abelian group $A$ is a subgroup of $G_P$, this decomposition serves as a valid Cartan decomposition of $g_\sigma(x_i)$ in $G$ as well. That is, the Cartan decomposition of $G_P$ defines a particular gauge fixing of Cartan factors for the family of group elements $g_\sigma(x_i)$. The rotations in $G_P$ Lie in the $U(1)$ group generated by~$M_{12}$. These rotations commute with the Weyl inversion $w$ when $d>2$. Hence the factors $\kappa_i$ that arise in the transition from $s$- to $t$-channel must be of the form
\begin{equation*}
 \kappa_i = {\rm e}^{\gamma_i D} {\rm e}^{\varphi_i M_{12}} ,
\end{equation*}
for some functions $\gamma_i$ and $\varphi_i$ that depend on the two cross ratios constructed out of insertion points $x_i$. A direct calculation gives
\begin{equation*}
 \kappa_1 = \kappa_3 = {\rm e}^{\gamma D + \alpha M_{12}},\qquad \kappa_2 = \kappa_4 = {\rm e}^{\gamma D - \alpha M_{12}} ,
\end{equation*}
 with
\begin{equation*} {\rm e}^{4\gamma} = \frac{x_{12}^2 x_{34}^2}{x_{14}^2 x_{23}^2},\qquad
 {\rm e}^{2i\alpha} = \frac{\cosh{\frac{u_1}{2}}}{\cosh{\frac{u_2}{2}}} .
\end{equation*}
The coordinates on $A$ are related to cross ratios of section 2 by $\sinh^{-2}\frac{u_i}{2} = z_i$. We have performed the calculation after moving the points to a configuration
\begin{gather*}
 x_1 = \frac{\cosh^2\frac{u_1}{2}+\cosh^2\frac{u_2}{2}}{2\cosh^2\frac{u_1}{2}\cosh^2\frac{u_2}{2}} e_1 - {\rm i}
 \frac{\cosh^2\frac{u_1}{2}-\cosh^2\frac{u_2}{2}}{2\cosh^2\frac{u_1}{2}
 \cosh^2\frac{u_2}{2}} e_2 ,\\ x_2 = 0 ,\qquad x_3 = e_1 ,\qquad x_4 = \infty e_1 .
\end{gather*}
Although $\mathcal{M}$ was originally defined using representations of $K = {\rm SO}(1,1)\times {\rm SO}(d)$, it is computed using only representation theory of ${\rm SO}(1,1)\times {\rm SO}(2)$.

Let us elaborate on the last comment by looking at theories in $d=3$ dimensions. Following~\cite{Buric:2019dfk} we parametrise the elements $r$ of the 3-dimensional rotation group through Euler angles,
\begin{equation*}
 r(\phi,\theta,\psi) = {\rm e}^{-\phi M_{12}} {\rm e}^{-\theta M_{23}} {\rm e}^{-\psi M_{12}}.
\end{equation*}
With this choice of coordinates, the elements $\kappa_i$ have $\phi =\pm\alpha$ and
$\theta = \psi = 0$. Next let us recall that matrix elements of the spin-$j$
representation of ${\rm SU}(2)$ read
\begin{equation*}
t^j_{m n} (\phi,\theta,\psi) = \langle j,m| g(\phi,\theta,\psi) | j,n\rangle =
{\rm e}^{-{\rm i}(m\phi+n\psi)} d^j_{m n}(\theta).
\end{equation*}
Here, the function $d^j_{m n}$ is known as Wigner's $d$-function. It is expressed
in terms of Jacobi polynomials $P^{(\alpha,\beta)}_n$ as
\begin{equation*}
d^j_{m n}(\theta) = {\rm i}^{m-n} \sqrt{\frac{(j+m)!(j-m)!}{(j+n)!(j-n)!}}
\left(\sin\frac{\theta}{2}\right)^{m-n} \left(\cos\frac{\theta}{2}\right)^{m+n}P^{(m-n,m+n)}_{j-m}(\cos\theta).
\end{equation*}
For $\theta=0$, the only non-zero matrix elements are those with $m=n$. Furthermore
\begin{equation*}
t^j_{n n}(\pm\alpha,0,0) = {\rm e}^{\mp {\rm i}n\alpha} P^{(0,2n)}_{j-n}(1) =
{\rm e}^{\mp {\rm i}n\alpha} = \left(\frac{\cosh\frac{u_1}{2}}{\cosh\frac{u_2}{2}}\right)^{\mp\frac{n}{2}}.
\end{equation*}
Since the stabiliser group $B = {\rm SO}(d-2)$ for a bosonic conformal field theory in $d=3$ dimensions is trivial, so taking $B$-invariants is trivial. Putting all this together we conclude that the crossing factor reads
\begin{equation*}
 \mathcal{M}^I_{\ J} = \mathcal{M}^{ijkl}_{pqrs} = \left(\frac{u}{v}\right)^{-\frac14\sum\Delta_i} \left(\frac{\cosh\frac{u_1}{2}}{\cosh\frac{u_2}{2}}\right)^{\frac12(i+k-j-l)} \delta^i_p \delta^j_q \delta^k_r \delta^l_s,
\end{equation*}
where $u$, $v$ are the usual $s$-channel cross ratios. Indices $i,p$ run through a basis of the representation space $V_1$ of $K$, $j$, $q$ through a basis of $V_2$ etc. The first factor in this result for the spinning crossing factor is well known from scalar correlators. For spinning correlators, it gets multiplied by a diagonal matrix whose entries are integer powers of ${\rm e}^{2{\rm i}\alpha}$. The analysis of this section can be repeated to include other channels. The argument given in the beginning of this subsection still goes through~-- one can compute the crossing factor from a two-dimensional theory. This again leads to the above structure of the crossing factor. We have done the computation of the crossing between $s$- and $u$-channels, where $\cosh{u_i/2}$ is replaced by~$\sinh{u_i/2}$ (and the scalar prefactor is modified appropriately).

\subsection{Example}\label{section4.4}

We can now put all the above together and compute the crossing factor between the $s$- and the $t$-channel for the $\mathcal{N} = 2$ superconformal algebra in one dimension. To this end, the first step
is to find the group elements $g_s(x_i)$ and $g_t(x_i)$ which appear in the argument of the covariant function $F$. In the previous section, we provided formulas for all ingredients that make up these elements.
Even for the simple example at hand, writing $g_s$ and $g_t$ as $3\times 3$ matrices whose coefficients are functions in all the $u_i$, $\theta_i$, $\bar\theta_i$ for $i=1, \dots, 4$ is rather cumbersome. At this point, the superconformal invariance comes to our rescue, as it allows to move the four points to the special position
\begin{equation} \label{eq:xigauge}
 x_1 = \big(x,\theta_1,\bar\theta_1\big),\qquad x_2 = (0,0,0),\qquad x_3 = \big(1,\theta_3,\bar\theta_3\big),\qquad x_4 = (\infty,0,0) .
\end{equation}
With this gauge choice, the entries of the matrices $g_s(x_i)$ and $g_t(x_i)$ depend on the bosonic coordinate~$x$ and the four Grassmann variables $\theta_1$, $\theta_3$ and $\bar \theta_1$, $\bar\theta_3$ only.

Next, we need to find the Cartan decomposition of the elements $g_s$ and $g_t$. The Cartan coordinates on ${\rm SL}(2|1)$ are introduced by
\begin{gather*}
g={\rm e}^{\kappa R}{\rm e}^{\lambda_l D}{\rm e}^{\bar q Q_-+\bar s S_-}{\rm e}^{\frac{u}{2}(P+K)}{\rm e}^{qQ_++sS_+}{\rm e}^{\lambda_r D} .
\end{gather*}
This agrees with the general prescription \eqref{eq:sCartan}, except that the abelian factor $A$ is one-dimen\-sio\-nal rather than two-dimensional. The elements $g_s$, $g_t$ and their Cartan coordinates are found by simple multiplication of supermatrices. The bosonic Cartan coordinates in $s$-channel read
\begin{gather*}
 \cosh^2 \frac{u_s}{2} = \frac1x \left( 1 - \frac12\theta_3\bar\theta_3 -\frac{\theta_1\bar\theta_1}{2x} + \frac{\theta_1\bar\theta_3}{x} +\frac{\theta_1\bar\theta_1\theta_3\bar\theta_3}{4x} \right),\\
 {\rm e}^{\lambda_{s,l}-\lambda_{s,r}} = \left(1-x-\frac12\theta_1\bar\theta_1-\frac12\theta_3\bar\theta_3+\theta_1\bar\theta_3\right) \left(x-\frac12\theta_1\bar\theta_1\right),\\
 {\rm e}^{\lambda_{s,l}+\lambda_{s,r}} \left(1+\frac12\theta_3\bar\theta_3\right)\left(x-\frac12\theta_1\bar\theta_1\right), \qquad {\rm e}^{-2\kappa_s} = 1 + \frac{\theta_1}{x}(\bar\theta_1 - \bar\theta_3) .
\end{gather*}
In the $t$-channel, they are
\begin{gather*}
 \cosh^2 \frac{u_t}{2} = x\left( 1 + \frac12\theta_3\bar\theta_3 + \frac{\theta_1\bar\theta_1}{2x} -\theta_1\bar\theta_3 + \frac{\theta_1\bar\theta_1\theta_3\bar\theta_3}{4x}\right),\\
 {\rm e}^{\lambda_{t,l}-\lambda_{t,r}} = - \left(1-x-\frac12\theta_1\bar\theta_1-\frac12\theta_3\bar\theta_3+\theta_1\bar\theta_3\right)\left(1+\frac12\theta_3\bar\theta_3\right),\\
 {\rm e}^{\lambda_{t,l}+\lambda_{t,r}} = \left(1-\frac12\theta_3\bar\theta_3\right)\left(x+\frac12\theta_1\bar\theta_1\right), \qquad {\rm e}^{-2\kappa_t} = 1 + \bar\theta_3 (\theta_3 - \theta_1) .
\end{gather*}
The fermionic Cartan coordinates, on the other hand, are given by the following expressions
\begin{gather*}
q_s = {\rm e}^{\frac12\lambda_{s,r}}\left(\theta_3 -\frac{\theta_1}{x} \left( 1-\frac12\theta_3\bar\theta_3\right)\right),\qquad s_s = {\rm e}^{-\frac12\lambda_{s,r}}\frac{\theta_1}{x},\\
 \bar q_s = {\rm e}^{-\frac12\lambda_{s,l}}\big(\bar\theta_3 -\bar\theta_1\big),\qquad \bar s_s = -{\rm e}^{\frac12\lambda_{s,l}}\frac{\bar\theta_3}{x},\\
 q_t = {\rm e}^{\frac12\lambda_{t,r}}(\theta_3 - \theta_1),\qquad s_t = -{\rm e}^{-\frac12\lambda_{t,r}}\theta_1\left(1-\frac12\theta_3\bar\theta_3\right),\\
 \bar q_t = -{\rm e}^{-\frac12\lambda_{t,l}}\left(\bar\theta_1 - \bar\theta_3\left(x+\frac12\theta_3\bar\theta_1\right)\right),\qquad \bar s_t = {\rm e}^{\frac12\lambda_{t,l}}\bar\theta_3 .
\end{gather*}
Finally, using these expressions, we can compute crossing factor $\mathcal{M}_{\rm st}$ between the two channels. For the superconformal algebra $\mathfrak{sl}(2|1)$ the group $K$ is generated by dilations $D$ and $R$-symmetry transformations~$R$. Therefore, it is abelian, so all its irreducible representations are 1-dimensional. The factor $\mathcal{M}_{\rm st}$ is hence just a single function in the variables $x$, $\theta_1$, $\theta_3$ and~$\bar \theta_1$,~$\bar\theta_3$. It depends, of course, on the choice of representations $(\Delta_i,r_i)$ for the external superfields. Note that in our gauge~\eqref{eq:xigauge} the factors $k(t_{41})$ and $k(t_{43})$ are trivial. Therefore, we have
\begin{alignat*}{3}
& \kappa_1 = {\rm e}^{(\lambda_{s,l}-\lambda_{t,l})D + (\kappa_s - \kappa_t) R} ,\qquad&&
 \kappa_4 = {\rm e}^{(\lambda_{t,l}+\lambda_{s,r})D-\kappa_t R},&\\
& \kappa_3 = {\rm e}^{(\lambda_{t,r}-\lambda_{s,r})D} ,\qquad &&\kappa_2 =
 {\rm e}^{-(\lambda_{t,r}+\lambda_{s,l}-\log x^2)D + (\kappa_s - \frac12\theta_3\bar\theta_3 +
 \frac{\theta_1\bar\theta_1}{2x})R}.&
\end{alignat*}
Therefore, the Cartan coordinates from above yield the following expression for $\mathcal{M}_{\rm st}$
\begin{gather*}
 \mathcal{M}_{\rm st}  = {\rm e}^{\frac{{\rm i}\pi}{2}(\Delta_2+\Delta_4-\Delta_1-\Delta_3)} x^{-2\Delta_1}\alpha^{\frac32\Delta_1-\frac12\Delta_2-\frac12\Delta_3-\frac12\Delta_4} \\
\hphantom{\mathcal{M}_{\rm st}  =}{} \times \beta^{\frac12\Delta_1+\frac12\Delta_2-\frac32\Delta_3+\frac12\Delta_4}{\rm e}^{r_1(\kappa_s-\kappa_t) +r_2(\kappa_s-\frac12\theta_3\bar\theta_3+\frac{\theta_1\bar\theta_1}{2x})-r_4\kappa_t},
\end{gather*}
where $\alpha$ and $\beta$ are defined by
\begin{equation*}
 \alpha = x + \tfrac12\theta_1\bar\theta_1,\qquad \beta = 1 - \tfrac12\theta_3 \bar\theta_3 .
\end{equation*}
Before ending this section, let us mention that, in order to analyse crossing equations, one would expand functions $f(x,\theta_i,\bar\theta_i)$ in Grassmann variables and restrict to $B$-invariants, which, in the case at hand are $R$-invariants. In this process, the factor~$\mathcal{M}$ is turned to a $6\times6$ matrix of differential operators. Details on this point are given in~\cite{Buric:2020buk}.

\section{Casimir equations and their solution}\label{section5}

The aim of this section is to study the Casimir equations for superconformal partial waves. We will review the method of computing these functions that was introduced in~\cite{Buric:2019rms}. However, let us first give a short account on previous works on superconformal blocks. We will then state which open problems our method is supposed to address.

Superconformal partial waves have been so far computed in a number of examples, for various spacetime dimensions and types of correlators. There are several techniques of deriving them, which include making an ansatz as a sum of bosonic blocks and fixing coefficients using Ward identities \cite{Chang:2017xmr,Dolan:2004mu,Gimenez-Grau:2020jrx,Lemos:2016xke,Li:2016chh, Liendo:2016ymz,Liendo:2018ukf,Nirschl:2004pa,Poland:2010wg}, directly evaluating shadow integrals \cite{Fitzpatrick:2014oza} and solving the appropriate Casimir differential equations \cite{Bissi:2015qoa,Bobev:2015jxa,Bobev:2017jhk, Khandker:2014mpa,Lemos:2015awa}. What is common to these works is that they either focus on correlators of short operators so that the blocks have only one component, or otherwise restrict to the bottom component of blocks (by setting all Grassmann variables to zero). The bottom component takes the form of a finite sum of {\it scalar} bosonic blocks. It is desirable to extend the analysis and derive all components of superconformal blocks (termed {\it long blocks}) because they lead to a larger set of crossing equations for the same OPE data. First steps in this direction have been performed in \cite{Cornagliotto:2017dup,Gimenez-Grau:2019hez,Kos:2018glc} where certain long blocks in 1-dimensional problems have been derived. As far as we know, the only computation of long blocks that is not ultimately reducible to one dimension was done in~\cite{Ramirez:2018lpd}. It is expected that long blocks in general should be expressible as finite sums of spinning bosonic blocks.\footnote{There is another, very general but less explicit, approach to superconformal partial waves proposed in~\cite{Doobary:2015gia}.}

The construction of~\cite{Buric:2019rms} that we will describe address the above issues in two respects. First, it produces the Casimir equations for any kind four-point functions in theories with type I superconformal symmetry in terms of Casimir equations for appropriate spinning bosonic partial waves. This is possible due to the fact that under the map~(\ref{magic-formula}) the Casimir equations are carried to the eigenvalue problem for the Laplacian~$\Delta$ on the space of $K$-spherical functions. In general, the expression for the Laplacian on a type I supergroup can be related to its bosonic counterpart~$\Delta_0$ provided that one works in a particular coordinate system that we will introduce in the second subsection. In these coordinates, $\Delta$ differs from~$\Delta_0$ by a nilpotent term~$A$ and a~simple ``trace'' term.

Next, we will analyse the eigenfunctions of the Laplacian. The term $A$ will be treated as a perturbation, so that the eigenfunctions of $\Delta$ are obtained from those of $\Delta_0$ by means of (quantum mechanical) perturbation theory that terminates at a finite order. This procedure will be explained in the third subsection. The eigenproblem of~$\Delta_0$ itself will be reviewed in the first subsection. It reduces to a two-particle Schr\"odinger problem that, in the case of scalar fields, coincides with the $BC_2$ Calogero--Sutherland system,~\cite{Schomerus:2016epl}. For spinning fields, the equation is already non-trivial but its solutions, the spinning bosonic partial waves, have been extensively studied in the literature. Thus, the upshot of the whole construction is to obtain superconformal blocks in a systematic way from well known functions.

Finally, as in the previous sections, we will apply the general theory in the $\mathfrak{sl}(2|1)$ example. For a more involved example in four dimensions, the reader is referred to~\cite{N1D4_paper}.

\subsection{Casimir equations and Calogero--Sutherland models}\label{section5.1}

The Laplace--Beltrami operator $\Delta$ on a Lie group $G$ may be constructed as the quadratic Casimir build out of left-invariant vector fields. Alternatively, one may use right-invariant vector fields. The two operators obtained in this way coincide. The Laplacian is a second order differential operator acting on the algebra of functions~$C^\infty(G)$. More generally, $\Delta$ can act on vector-valued functions component-wise.

As invariant vector fields form a representation of the Lie algebra $\mathfrak{g} = {\rm Lie}(G)$, the Laplacian commutes with them. It follows that under $\Delta$ the space of $K$-spherical functions is mapped to itself. The space of these functions will be denoted by
\begin{gather*}
 \Gamma^G_{V_l,V_r} = \big\{ F\colon G\xrightarrow{}V_l\otimes V_r\, |\, F(k_l g k_r)= \big(\rho_l(k_l)\otimes\rho_r(k_r)^{-1}\big) F(g)\big\}.
\end{gather*}
Thus, $\Gamma$ is specified by two finite-dimensional representations $\rho_l$, $\rho_r$ of $K$, on spaces $V_l$, $V_r$. Due to the Cartan decomposition $G=KAK$, any function in $\Gamma_{V_l,V_r}$ is uniquely determined by the values it assumes on the two-dimensional abelian group~$A$. Not every function $f\colon A\xrightarrow{}V_l\otimes V_r$ can be extended to a $K$-spherical function, because of the non-uniqueness of the decomposition. Only functions which take values in the space of invariants $(V_l\otimes V_r)^B$ admit consistent extensions.

In any coordinate system on $G$, we can find the expression for the Laplacian, e.g., by first computing the Maurer--Cartan form ${\rm d}g \, g^{-1} = {\rm d}x_a\, C_{ab} X^b$. The right invariant vector fields are then
\begin{equation*}
 \mathcal{R}_a = \mathcal{R}_{X^a} = C^{-1}_{ab}\partial_b .
\end{equation*}
Here $a=1,\dots,\dim \mathfrak{g}$ runs over the basis $\{X^a\}$ of the conformal Lie algebra and $(x_a)$ are any local coordinates on~$G$. The Laplacian is then found with the help of the Killing form~$K^{ab}$ as $\Delta = K^{ab}\mathcal{R}_a\mathcal{R}_b$.

Working in Cartan coordinates makes it particularly easy to restrict~$\Delta$ to the space of $K$-spherical functions. The resulting operator may be regarded as acting on the space of functions
\begin{equation*}
 \Gamma^A_{V_l,V_r} = \big\{ f \colon A \xrightarrow{} (V_l\otimes V_r)^B \big\} .
\end{equation*}
The Laplacian is self-adjoint with respect to the scalar product on $L^2(G)$ that uses the Haar measure. One has to take this into account in order to obtain an operator of the Schr\"odinger form on~$\Gamma^A$. This is achieved by conjugation with a scalar factor~$\omega$~\cite{Buric:2019dfk}
\begin{equation*}
 \omega(u_1,u_2) = 4(-1)^{2-d}\left(\sinh\frac{u_1}{2} \sinh\frac{u_2}{2}\right)^{2d-2}\coth\frac{u_1}{2}\coth\frac{u_2}{2} \left|\sinh^{-2}\frac{u_1}{2}-\sinh^{-2}\frac{u_2}{2}\right|^{d-2} .
\end{equation*}
More precisely, we set
\begin{equation*}
 H_{\rho_l,\rho_r} = 2\omega^{1/2} \Delta_A \omega^{-1/2} - \frac14 (d-1)^2 .
\end{equation*}
The reader is referred to~\cite{Schomerus:2017eny} for details, which we will not need in the following. To finish this subsection, we quote the Hamiltonian that is obtained in the above process in the case of four scalar fields.
The Hamiltonian~$H$ takes the form
\begin{equation}\label{Hamiltonian-reduction}
 H_{\rho_l,\rho_r} = - \frac{\partial^2}{\partial u^2_1} - \frac{\partial^2}{\partial u^2_2} + V^{\rm CS}_{\rho_l,\rho_r}(u_1,u_2),
\end{equation}
where $V$ is the potential that depends on the representations $\rho_l,\rho_r$ defining covariance laws as
\begin{gather*}
V_{\rm CS}^{\hat \pi,s} (u_i)= V^{(a,b,\epsilon)}_{\rm CS}(u_i) = V_{\rm PT}^{(a,b)}(u_1)+V_{\rm PT}^{(a,b)}(u_2)+\frac{\epsilon(\epsilon-2)} {8\sinh^2\frac{u_1-u_2}{2}}+\frac{\epsilon(\epsilon-2)}{8\sinh^2\frac{u_1+u_2}{2}},\\
V_{\rm PT}^{(a,b)}(u)=\frac{(a+b)^2-\frac{1}{4}}{\sinh^2u}-\frac{ab}{\sinh^2\frac{u}{2}} .
\end{gather*}
Here the parameters $a$ and $b$ are conformal weights of $\rho_l$, $\rho_r$, respectively, and $\epsilon= d-2$. The one-dimensional potential $V_{\rm PT}$ is known as the P\"oschl--Teller potential and the Hamiltonian~(\ref{Hamiltonian-reduction}) is that of the $BC_2$ Calogero--Sutherland system. The derivation of~(\ref{Hamiltonian-reduction}) along with the computation of the potentials for a number of examples can be found in \cite{Buric:2019rms, Schomerus:2017eny,Schomerus:2016epl}.

\subsection{Laplacian on type I supergroups}\label{section5.2}

Having a good control over bosonic Casimir equations, we now move to the super-case. For the precise meaning of phrases such as ``supergroup'' or ``space of functions on a supergroup'' the reader is referred to appendices. In the following we will use such phrases in a somewhat loose, but hopefully clear, way.

As a vector space, the space of functions on a supergroup $G$ is isomorphic to the space of vector valued function on the underlying Lie group
\begin{equation*}
 C^\infty(G)\cong C^\infty\big(G_{(0)},\Lambda\mathfrak{g}_{(1)}^\ast\big) .
\end{equation*}
The vector space in which the functions take values is dual to the exterior algebra on the odd part of $\mathfrak{g}={\rm Lie}(G)$. This corresponds to the expansion of a function in Grassmann coordinates. Similarly, vector valued functions on $G$ may be regarded as function from the bosonic group~$G_{(0)}$ to the space tensored with $\Lambda\mathfrak{g}_{(1)}^\ast$
\begin{equation*}
 C^\infty(G,V)\cong C^\infty\big(G_{(0)},V\otimes\Lambda\mathfrak{g}_{(1)}^\ast\big) .
\end{equation*}
The Laplacian on the supergroup commutes with left and right invariant vector fields and therefore acts within the space of $K$-spherical functions. There is a simple relation between the Laplacian on~$G$ and its bosonic counterpart on~$G_{(0)}$ when the superalgebra~$\mathfrak{g}$ is of type~I. To review this relation, we consider a slight modification of the Cartan coordinates that we analysed in the previous section (see~\cite{Quella:2007hr})
\begin{equation} \label{param}
g = \eta'_l k_l a k_r \eta_r' = {\rm e}^{\bar\sigma^a \bar Y_a} k_l a k_r {\rm e}^{\sigma^a Y_a} .
\end{equation}
That is, we have commuted the factors that include fermionic coordinates past the factors $k_l, k_r\in K$ to place them on the furthest left and right positions. That $k_l$ and $k_r$ remain unchanged in this process follows from the Baker-Campbell-Hausdorff formula. The elements~$\bar Y_a$ form a basis of $\mathfrak{g}_-$ with the index $a$ running through $a = 1, \dots, \dim \mathfrak{g}_-$. Elements of the dual basis in~$\mathfrak{g}_+$ are denoted by~$Y_a$. When written in terms of the supercharges and special superconformal transformations, the exponents read
\[ \bar\sigma^a \bar Y_a = \bar\sigma^\beta_q Q^-_\beta + \bar \sigma_s^\beta S^-_\beta
, \qquad \sigma^a Y_a = \sigma^\beta_q Q^+_\beta + \sigma_s^\beta S^+_\beta . \]
Here $Q^\pm_\beta$ is a basis of $\mathfrak{q}_\pm$ and $S^\pm_\beta$ is a basis of
$\mathfrak{s}_\pm$ so that $\beta$ runs through $\beta = 1, \dots, \dim
\mathfrak{g}_{(1)}/4$. Since we have moved the fermionic generators from $\mathfrak{g}_-$into the leftmost factor, left translations with elements $k\in K$ act on the corresponding Grassmann coordinates.
Hence, the covariance laws satisfied by $K$-spherical functions do not hold component-wise after the expansion in Grassmann variables. They mix various components of the vector-valued functions. The mixing can be expressed by saying that the expanded $K$-spherical functions on $G$ naturally correspond to $K$-spherical functions on the bosonic group~$G_{(0)}$, belonging to the space $\Gamma^{G_{(0)}}_{V_l ,V_r}$ with
\begin{equation}\label{super-reps}
 V_l = V_{(12)} \otimes \Lambda\mathfrak{g}_-^\ast , \qquad V_r = V_{(34)}\otimes\Lambda\mathfrak{g}_+^\ast.
\end{equation}
The most important property of the coordinates defined above is the form that the Laplacian assumes in them. Namely, we have \cite{Gotz:2006qp, Quella:2007hr}
\begin{equation*}
 \Delta = \Delta_0 - 2 D^{ab}\bar \partial_{\sigma^a} \partial_{\bar\sigma^b} - K_{ij} \operatorname{tr}\big(D\big(X^i\big)\big)\mathcal{R}^{(0)}_{X^j} .
\end{equation*}
Here, $\Delta_0$ is the Laplacian on the bosonic subgroup of $G$. For a superconformal group this Casimir of the bosonic subgroup receives a very simple correction: a term that involves only second order derivatives of fermionic coordinates with bosonic coefficients. The coefficients $D^{ab}$ are matrix elements of the representation $D$ of the bosonic group $G_{(0)}$ on the space $\mathfrak{g}_+$, restricted to the section
$A = A_{(0)} \subset G_{(0)}$ of the bosonic conformal group. Since the functions $D^{ab}$ depend only on the bosonic coordinates, the correction term is a nilpotent operator. This is clear upon expansion of functions in Grassmann variables. In the final term, $\{X^i\}$ denotes any basis of the even subalgebra $\mathfrak{g}_{(0)}$ and~$K_{ij}$ the Killing form in this basis. This term has only one non-zero contribution, coming from the $U(1)$ $R$-charge generator, which is the only one that is not traceless in the representation~$D$.

The process of reduction proceeds along the same lines as explained in the previous section. Therefore, we end up with a matrix-valued Hamiltonian
\begin{equation*}
 H = H'_0 + A .
\end{equation*}
Here, $H_0$ is the spinning bosonic Hamiltonian specified by representations~(\ref{super-reps}). It is modified to~$H'_0$ by addition of constants along the diagonal, which come from the trace term. Finally, $A$~is a nilpotent matrix obtained from the reduction of the second term.

\subsection{Nilpotent perturbation theory}\label{section5.3}

Having seen that for type I superconformal symmetry the Casimir operator differs form the spinning bosonic one by a nilpotent piece $A$, our strategy is to treat~$A$ as a perturbation and construct supersymmetric partial waves as a perturbation of spinning partial waves. Since~$A$ is nilpotent, the process will produce exact results at some finite order $N\leq\dim \mathfrak{g}_{(1)}/2$. General methods to solve for eigenfunctions of a Hamiltonian $H= H_0 + A$ in terms of those of $H_0$ are well established. Particularly suited for our purposes is the exposition of Messiah,~\cite{messiah1962quantum}, that we now review. For simplicity, we assume that~$H$ and $H_0$ have discrete spectra and finite dimensional eigenspaces. By a limiting process, the construction can be extended to more general spectra.

Let us first set up a bit of notation. We will write $H_{[0]}$ to mean either $H_0$ or $H$ in order to avoid cluttering. Similar remarks apply to all objects that will carry such indices. The Hilbert space on which the operators act is denoted by $\mathcal{H}$ and $H_0$ is assumed to be hermitian. We shall denote the eigenspaces of $H_{[0]}$ by $V^{[0]}_n$ and the corresponding eigenvalues by $\varepsilon^{[0]}_n$. Projectors to these eigenspaces are written as $P^{[0]}_n$. Consider the resolvents
\begin{equation*}
G_{[0]}\colon \ \mathbb{C}\xrightarrow{}L(\mathcal{H}) ,\qquad G_{[0]}(z) = \big(z-H_{[0]}\big)^{-1}.
\end{equation*}
They can be expanded in the projectors $P_n^{[0]}$ with simple poles at the eigenvalues $\varepsilon_n^{[0]}$ of $H_{[0]}$. Conversely, the projectors are obtained by picking up the residues of resolvents at the position of eigenvalues
\begin{equation*}
G_{[0]}(z) = \sum_n \frac{1}{z-\varepsilon^{[0]}_n} P^{[0]}_n , \qquad P^{[0]}_n = \frac{1}{2\pi i}\oint_{\Gamma_n} G_{[0]}(z)\, {\rm d}z.
\end{equation*}
Here $\Gamma_n$ is a small contour encircling $z=\varepsilon^{[0]}_n$ and none of the other eigenvalues.

Let us insert the relation $H = H_0 + A$ between the two Hamilton operators into the resolvent~$G$. Upon expanding in~$A$ we get
\begin{equation*}
G = G_0 \sum_{n=0}^{\infty} (A G_0)^n = G_0 \sum_{n=0}^{N} (A G_0)^n .
\end{equation*}
The infinite sum truncates at a finite order for the kind of operators we wish to consider. For, if $A^N = 0$ then also $(AG_0)^N=0$, since in our application $G_0$ acts diagonally on $\mathcal{H} = L^2(\mathbb{C}^m)\otimes\mathbb{C}^l$ and $A$ is a triangular matrix of functions. Computing residues of the previous expansion for $G$ at $\varepsilon^0_i$ we obtain
\begin{equation}
P_i = P^0_i + \sum_{n=1}^{N} {\rm Res}\big(G_0 (AG_0)^n,\varepsilon^0_i\big)\equiv P^0_i + P^{(1)}_i + \dots + P^{(N)}_i, \label{proj}
\end{equation}
with
\begin{gather}
P^{(1)}_i = P^0_i A S_i + S_i A P^0_i , \label{eq:P1}\\
P^{(2)}_i = P^0_i A S_i A S_i + S_i A P^0_i A S_i + S_i A S_i A P^0_i\nonumber\\
\hphantom{P^{(2)}_i =}{} -P^0_i A P^0_i A S_i^2 - P^0_i A S_i^2 A P^0_i - S_i^2 A P^0_i A P^0_i,\label{eq:P2}
\end{gather}
and so on. Here, the symbol $S_i$ denotes the following operator
\begin{equation*}
S_i = \sum_{j\neq i}\frac{P^0_j}{\varepsilon^0_i - \varepsilon^0_j}. 
\end{equation*}
Since the sum over $j$ is restricted to $j \neq i$ we infer that $S_i P^0_i = P^0_i S_i = 0$, a property we shall frequently use. The idea now is to find the eigenvectors of $H$ by applying the projections $P_i$ to eigenvectors of $H_0$. This leads to the complete solution of the problem provided that the maps, $P_i\colon V^0_i\xrightarrow{}V_i$ and $P^0_i\colon V_i\xrightarrow{}V^0_i$ are vector space isomorphisms. We shall verify that this assumption is true in the example below.

\subsection{Example}\label{section5.4}

In this final subsection, we will apply the nilpotent perturbation theory to find the partial waves for the $\mathcal{N}=2$ superconformal symmetry in one dimension. The even subalgebra of the superconformal algebra $\mathfrak{g} = \mathfrak{sl}(2|1)$ is $\mathfrak{g}_{(0)}=\mathfrak{so}(1,2)\oplus \mathfrak{u}(1)$. Its representations $[j,q]$ are labelled by a spin $j$ and an $R$-charge~$q$. For finite dimensional (non-unitary) representations, $j$ is half-integer while $q$ can be any complex number. We see that the odd subspace $\mathfrak{g}_{(1)}$ decomposes into a sum of two irreducible representations,
\begin{gather*}
 \mathfrak{g}_{(1)} = \mathfrak{g}_+ \oplus \mathfrak{g}_- = [1/2,1]\oplus [1/2,-1] .
\end{gather*}
When we restrict the representations $\mathfrak{g}_\pm$ to the subalgebra $\mathfrak{k} = \mathfrak{u}(1)_D\oplus\mathfrak{u}(1)_R$, they decompose into a sum of two irreducibles each
\begin{equation*}
\mathfrak{g}_+ = \mathfrak{q}_+ \oplus \mathfrak{s}_+  , \qquad \mathfrak{g}_- = \mathfrak{q}_- \oplus \mathfrak{s}_- ,
\end{equation*}
where
\begin{equation*}
\mathfrak{q}_\pm = (1/2,\pm 1) , \qquad \mathfrak{s}_\pm = (-1/2,\pm 1) .
\end{equation*}
The first label of the representation is the conformal weight $\Delta$, while the second is the $R$-charge~$q$. Recall that $\mathfrak{q}_\pm$ are the spaces spanned by $Q_\pm$, respectively, and
the same for $\mathfrak{s}_\pm$. In our analysis of the Casimir equations it is important to know the representation content of $\Lambda\mathfrak{g}_\pm$ which is given by
\begin{equation*}
\Lambda\mathfrak{g}_\pm = (0,0) \oplus (1/2,\pm 1) \oplus (-1/2,\pm 1) \oplus (0,\pm 2) .
\end{equation*}
The Lie superalgebra $\mathfrak{sl}(2|1)$ possesses two algebraically independent Casimir elements, one of second order and one of third. The quadratic Casimir element is given by
\begin{equation}
C_2=-D^2+\tfrac14 R^2 - \tfrac12 \{K,P\} +\tfrac12 [Q_+,S_-] - \tfrac12 [Q_-,S_+] . \label{e1}
\end{equation}
The cubic Casimir element, on the other hand, reads
\begin{gather*}
C_3  = \big(D^2-\tfrac14 R^2+PK\big)R - Q_+ S_- \big(D + \tfrac32 R\big)\\
\hphantom{C_3  =}{} - Q_-S_+ \big(D-\tfrac32 R\big)+ K Q_+ Q_- + P S_-S_+ - D - \tfrac12 R .
\end{gather*}
Typical representations of the superalgebra $\mathfrak{sl}(2|1)$ can be distinguished by the values of these two Casimir elements. For atypical representation (short multiplets) this is not the case. These representations have
both Casimirs are zero, see, e.g.,~\cite{Scheunert:1976wj}.

As explained in the previous subsection, we shall deviate from the Cartan coordinates and parametrise the supergroup as
\begin{equation*}
g = {\rm e}^{\bar\varrho Q_- + \bar\sigma S_-} {\rm e}^{\kappa R} {\rm e}^{\nu_1 D} {\rm e}^{\mu (P-K)} {\rm e}^{\nu_2 D}{\rm e}^{\varrho Q_+ - \sigma S_+} . 
\end{equation*}
We also used the generator $P-K$ in the middle factor in order to have discrete spectrum for which our discussion the perturbation theory readily applies. Solutions that we will obtain can be then analytically continued in parameters as well as the argument, as we will later show.

We can now perform the steps explained in the previous subsections to find the Laplacian and descend to the double coset $K\backslash G_{(0)} /K$. Since the algebra $\mathfrak{k}$ is abelian, the spaces~$V_i$ and hence also $V_{(12)}$ and $V_{(34)}$ are all one-dimensional. Recall that the Laplacian acts on a space of functions that take values in $B$-invariants. In the case at hand, $B$ coincides with the $R$-symmetry group~$U(1)$. We will assume that the $R$-charges $q_i$ of the four external fields sum up to $\sum q_i = 0$. Under this assumption, the space of $B$-invariant is 6-dimensional and spanned by
\begin{equation*}
\left(\Lambda \mathfrak{g}_{(1)}\otimes V_{(12)} \otimes V_{(34)}\right)^B
= \operatorname{span}\{1,\bar\sigma \sigma, \bar\sigma \varrho,\bar\varrho \sigma,\bar\varrho
\varrho,\bar\sigma \bar\varrho\sigma \varrho\} .
\end{equation*}
Each function on the one-dimensional coset space $K \backslash G_{(0)} /K$ that takes values
in this subspace can extended to a covariant function $f$ on the entire supergroup as
\begin{gather*}
f (\mu,\kappa,\nu,\sigma,\varrho) =
{\rm e}^{a\nu_1 + b\nu_2 + q\kappa} f_1 + {\rm e}^{(a+\frac12)\nu_1 + (b-\frac12)\nu_2 + (q+1)\kappa}f_{2}\, \bar\sigma\sigma \\
\hphantom{f (\mu,\kappa,\nu,\sigma,\varrho) =}{}
+ {\rm e}^{(a+\frac12)\nu_1 + (b+\frac12)\nu_2 + (q+1)\kappa}f_{3}  \bar\sigma\varrho
 + {\rm e}^{(a-\frac12)\nu_1 + (b-\frac12)\nu_2 + (q+1)\kappa}f_{4} \bar\varrho\sigma\\
\hphantom{f (\mu,\kappa,\nu,\sigma,\varrho) =}{}
  +{\rm e}^{(a-\frac12)\nu_1 + (b+\frac12)\nu_2 + (q+1)\kappa}f_{5} \bar\varrho\varrho + {\rm e}^{a\nu_1 + b\nu_2 + (q+2)\kappa} f_6 \bar\sigma\bar\varrho\sigma\varrho,
\end{gather*}
where the six real component functions $f_1,\dots,f_6$ depend on the variable $\mu$ that parametrises the double coset.

The behaviour of the individual terms under the left and right action of $K$ is determined by the parameters $(\Delta_i,q_i)$ of the external fields. Their values are $a = \Delta_2-\Delta_1$, $b= \Delta_3-\Delta_4$ and $q =q_1+q_2 = -q_3-q_4$. The precise form of the $\nu_i$ and $\kappa$-dependent prefactor depends on the fermionic coordinates they are multiplied with. The first term in the expansion above, one that contains no fermionic coordinates, is multiplied by the character of $K \times U(1)_D$ on $V_{(12)} \times V_{(34)}$ where $U(1)_D$ denotes the $U(1)$ subgroup of the right factor $K$ that is associated with dilation. In the remaining terms, this basic character is multiplied with the character of $K \times U(1)_D$ on the corresponding product of fermionic variables.

The Laplace--Beltrami operator is obtained by substituting explicit expressions for the left or right invariant vector fields in the quadratic Casimir $(\ref{e1})$. Applied to the function $f$ of the above form, it reduces to a second order differential operator in $\mu$ that acts on the vector $(f_1(\mu), \dots, f_6(\mu))$. The corresponding eigenvalue problem assumes the form of a matrix Schr\"odinger equation $H f = \lambda f$, with the Hamiltonian of the form $H = H_0 + A$ and
\begin{gather*} 
H_0 = \operatorname{diag} \bigg( H_{\rm PT}^{(a,b)}-\frac{(q-1)^2}{4} , H_{\rm PT}^{(a+\frac12,b-\frac12)}-\frac{q^2}{4} , H_{\rm PT}^{(a+\frac12,b+\frac12)}-\frac{q^2}{4},\\
\hphantom{H_0 = \operatorname{diag} \bigg(}{} H_{\rm PT}^{(a-\frac12,b-\frac12)}-\frac{q^2}{4},H_{\rm PT}^{(a-\frac12,b+\frac12)}-\frac{q^2}{4},H_{\rm PT}^{(a,b)}-\frac{(q+1)^2}{4}\bigg),
\end{gather*}
and a nilpotent perturbation
\begin{equation*}
A=
\begin{pmatrix}
0 & -\sin\mu & \cos\mu & -\cos\mu & -\sin\mu & 0\\
0 & 0 & 0 & 0 & 0 & \sin\mu\\
0 & 0 & 0 & 0 & 0 & -\cos\mu \\
0 & 0 & 0 & 0 & 0 & \cos\mu\\
0 & 0 & 0 & 0 & 0 & \sin\mu\\
0 & 0 & 0 & 0 & 0 & 0
\end{pmatrix}.
\end{equation*}
The unperturbed Hamiltonian $H_0$ contains six individual Hamiltonians $H_{\rm PT}^{(\alpha,\beta)}$ with a P\"oschl--Teller potential,
\begin{equation*}
H_{\rm PT}^{(\alpha,\beta)} = - \frac{1}{4}\partial_\mu^2 - \frac{\alpha\beta}{\sin^2 \mu} +
\frac{(\alpha+\beta)^2-\frac{1}{4}}{\sin^2 2\mu}.
\end{equation*}
The constants added to these P\"oschl--Teller Hamiltonians come from the trace term in the Laplacian $\Delta$. Let us now apply the nilpotent perturbation theory to solve the eigenvalue problem for $H$ in the case $a=b=q=0$. We have $A^3 = 0$, so the perturbation theory is exact at the second order. We shall solve the problem on the interval $\mu\in[0,\pi/2]$. The potential diverges at the boundaries and hence the spectrum is discrete.

According to our general discussion, we first need to spell out the solution of the unperturbed problem, i.e., provide the eigenfunctions of the P\"oschl--Teller Hamiltonians that appear along diagonal of $H_0$. Each of the operators
\[
H_{\rm PT}^{(0,0)}-1/4 ,\qquad H_{\rm PT}^{(1/2,-1/2)} = H_{\rm PT}^{(-1/2,1/2)} ,\qquad H_{\rm PT}^{(1/2,1/2)}= H_{\rm PT}^{(-1/2,-1/2)}
\]
has a unique eigenfunction that is non-singular on the above interval. They will be denoted by~$\psi_n$,~$\phi_n$ and~$\chi_n$ with $n = 0,1, \dots $ integer, respectively. Explicitly, we have
\begin{alignat*}{3}
& \psi_n = \sqrt{2(2n+1)}\sin^{1/2}\mu \cos^{1/2}\mu\ P_n(\cos2\mu) ,\qquad && \varepsilon^{0}_{0,n} = n(n+1),& \\
& \phi_n = 2\sqrt{n+1}\sin^{3/2}\mu \cos^{1/2}\mu\ P_n^{(1,0)}(\cos2\mu),\qquad && \varepsilon_{1,n}^0 = (n+1)^2,& \\
& \chi_n = 2\sqrt{n+1}\sin^{1/2}\mu \cos^{3/2}\mu\ P_n^{(0,1)}(\cos2\mu),\qquad && \varepsilon_{1,n}^0 = (n+1)^2.&
\end{alignat*}
Here, $P_n^{(\alpha,\beta)}$ denote Jacobi polynomials, $P_n = P_n^{(0,0)}$ are the Legendre polynomials. With this normalisation, each set of wave functions forms an orthonormal basis for the space of functions defined on the interval $[0,\pi/2]$ which vanish on the boundary, with respect to the usual scalar product,
\[ (g_1, g_2) = \int_0^{\frac{\pi}{2}} {\rm d}\mu\, g_1 (\mu) \bar g_2 (\mu),\]
for which the P\"oschl--Teller Hamiltonians are Hermitian. When we displayed the eigenvalues~$\varepsilon$ of the P\"oschl--Teller Hamiltonians we have already introduced the label $ i =
(\sigma,n), \sigma = 0,1$ that enumerates the various eigenspaces of the unperturbed Hamiltonian~$H_0$. We can now also display the associated projectors $P^0_i = P^0_{\sigma,n} = P^0_{\varepsilon_{\sigma,n}}$. They are given by
\begin{gather*}
 P^0_{n(n+1)} f = (\psi_n,f_1)\psi_n e_1 + (\psi_n,f_6)\psi_n e_6,\\
 P^0_{(n+1)^2}f = (\phi_n,f_2)\phi_n e_2 + (\phi_n,f_5)\phi_n e_5 + (\chi_n,f_3)\chi_n e_3 + (\chi_n,f_4)\chi_n e_4,
\end{gather*}
where $\{e_i\}$ is the standard orthonormal basis for $\mathbb{C}^6$ and $f=(f_1,\dots,f_6)^T$ is a six component column of functions in~$\mu$. In order to find eigenvectors of~$H$ we first need to compute the projectors~$P_i$. To do this, we need the following two integrals
\begin{gather*}
 I_1(m,n) = \int_0^{\pi/2} {\rm d}\mu\, \phi_m \psi_n\sin\mu = \sqrt{\frac{m+1}{2(2n+1)}}(\delta_{mn}-\delta_{m+1,n}),\\
 I_2(m,n) = \int_0^{\pi/2} {\rm d}\mu\, \chi_m \psi_n\cos\mu = \sqrt{\frac{m+1}{2(2n+1)}}(\delta_{mn}+\delta_{m+1,n}).
\end{gather*}
To evaluate the integrals, we performed the substitution to a new variable $x=\cos2\mu$ that takes values in $x\in[-1,1]$ and used the relations
\begin{equation*}
(1-x) P_n^{(1,0)} = P_n- P_{n+1} , \qquad (1+x) P_n^{(0,1)} = P_n + P_{n+1},
\end{equation*}
along with the orthogonality of Legendre polynomials. These results imply that
\begin{equation*}
P^0_i A P^0_i = 0,\qquad P^0_i A S_i A P^0_i=0 .
\end{equation*}
Here, the index $i = (\sigma,n)$ runs over $\sigma = 0,1$ and $n = 0,1,2, \dots$. To get the eigenvectors of $H$ all we have to do is to apply the projectors $P_i$ to $|\psi\rangle$. Using equation \eqref{proj} and the expressions~\eqref{eq:P1} and~\eqref{eq:P2} for $P^{(1)}_i$, $P^{(2)}_i$ we obtain the following set of linearly independent eigenfunctions of the perturbed Hamiltonian, i.e., the Laplacian on the supergroup,
\begin{gather*}
f^{(1)}_n = \psi_n e_1,\\
f^{(2)}_n = \phi_n e_2 - \frac{1}{\sqrt{2(n+1)(2n+1)}} \psi_n e_1 - \frac{1}{\sqrt{2(n+1)(2n+3)}} \psi_{n+1} e_1,\\
f^{(3)}_n = \chi_n e_3 + \frac{1}{\sqrt{2(n+1)(2n+1)}} \psi_n e_1 - \frac{1}{\sqrt{2(n+1)(2n+3)}} \psi_{n+1} e_1,\\
f^{(4)}_n = \chi_n e_4 - \frac{1}{\sqrt{2(n+1)(2n+1)}} \psi_n e_1 + \frac{1}{\sqrt{2(n+1)(2n+3)}} \psi_{n+1} e_1,\\
f^{(5)}_n = \phi_n e_5 - \frac{1}{\sqrt{2(n+1)(2n+1)}} \psi_n e_1 - \frac{1}{\sqrt{2(n+1)(2n+3)}} \psi_{n+1} e_1,\\
f^{(6)}_n = \psi_n e_6 + \frac{1}{\sqrt{2n(2n+1)}} (\phi_{n-1}(-e_2-e_5) +\chi_{n-1}(-e_3+e_4)) \\
\hphantom{f^{(6)}_n =}{} -\frac{1}{\sqrt{2(n+1)(2n+1)}} (\phi_n (e_2+e_5) +\chi_n (-e_3+e_4)) + \frac{2}{n(n+1)} \psi_n e_1.
\end{gather*}
Note that the superscript $(k)$ labels different solutions of our matrix Schr\"odinger equation. Each of the eigenfunctions $f^{(k)}$ has six components.

Let us make a couple of remarks about the obtained set of eigenfunctions. By completeness of eigenfunctions of each P\"oschl--Teller Hamiltonian, the eigenfunctions of $H_0$ are also complete in the Hilbert space of physical wave functions. However, the solution $f^{(6)}_n$ is not well-defined for $n=0$ and it is therefore discarded. Indeed, the perturbed Hamiltonian is seen to be no longer diagonalizable on the full Hilbert space, but it is diagonalizable on a codimension-one subspace. Non-diagonalizability is a known feature of the Laplacian on supergroups, \cite{Saleur:2006tf, Schomerus:2005bf}, and is related to the presence of atypical modules in the decomposition of the regular representation. In our case, as mentioned above, atypical (short) representations can appear only for eigenvalue zero, consistent with the findings here. In the conformal field theory language, the number of conformal blocks reduces when the field in the intermediate channel is BPS.

Conformal partial waves satisfy the same differential equations as the wavefunctions above, but different boundary conditions. We opted to work with the physical wavefunctions of the Schr\"odinger problem on a compact interval because the operator~$H_0$ was in this case manifestly Hermitian and had a discrete spectrum. Therefore, we could directly apply the perturbative procedure of the previous subsection. This means however, that a few further steps are needed to obtain the conformal blocks. Firstly, the solutions of the trigonometric model have to be adopted to the hyperbolic theory. This is done by expressing Jacobi polynomials in terms of the hypergeometric function $_2 F_1$, which allows to promote $n$ to a continuous parameter $\lambda$. Explicitly, let
\begin{equation*}
\Psi^{(a,b)}_\lambda = \left(\frac{4}{y}\right)^{a+\frac12}(1-y)^{\frac12 a-\frac12 b+\frac14} _2F_1\left(\frac12+a+\lambda,\frac12+a-\lambda,1+a-b,\frac{y-1}{y}\right),
\end{equation*}
where the variable $y$ is related to $u = 2{\rm i} \mu$ as $y = \cosh^{-2} \frac{u}{2}$. For our special values of the parameters $a$, $b$ we introduce in particular
\begin{equation*}
\Psi_\lambda = ({\rm i}\lambda)^{1/2} \Psi^{(0,0)}_\lambda ,\qquad
\Phi_\lambda = \frac12 ({\rm i}\lambda)^{3/2} \Psi^{(\frac12,-\frac12)}_\lambda ,\qquad
X_\lambda = \frac12 ({\rm i}\lambda)^{1/2} \Psi^{(\frac12,\frac12)}_\lambda. 
\end{equation*}
When these functions are specialised to integer or half-integer values of $\lambda$, we get the building blocks of the solution for the trigonometric model, more precisely
\begin{equation*}
\Psi_{\lambda=-n-\frac12} = \psi_{n} ,\qquad \Phi_{\lambda=-n-1} = \phi_{n} , \qquad X_{\lambda=-n-1} = \chi_{n} .
\end{equation*}
With this in mind, we define functions $F^{(i)}_\lambda$, $i=1,\dots,6$ by analytic continuation in $\lambda$ of the solutions $f^{(i)}_n$. These are solutions of the matrix Calogero--Sutherland model that are regular near the wall at $u=0$, in which incoming and outgoing waves are superposed in a very particular way. To extract the incoming and outgoing pieces, we decompose each $\Psi^{(a,b)}_\lambda$ as $\Psi^{(a,b)}_\lambda = \Psi^{(a,b)}_{\lambda,+} + \Psi^{(a,b)}_{\lambda,-}$, where
\begin{equation*}
\Psi^{(a,b)}_{\lambda,\pm} =c(\pm\lambda,a,b) 4^{\pm\lambda} (1-y)^{\frac12 a-\frac12 b +\frac14} y^{\mp \lambda}\, {} _2 F_1\left(\frac12+a\mp\lambda,\frac12-b\mp\lambda,1\mp2\lambda,y\right),
\end{equation*}
and the prefactor is $c$ given by
\begin{equation*}
c(\lambda,a,b) = 4^{-\lambda+a+\frac12} \frac{\Gamma(a-b+1) \Gamma(2\lambda)}{\Gamma\big(\frac12+\lambda+a\big)
	\Gamma\big(\frac12+\lambda-b\big)}.
\end{equation*}
Thus, the wavefunctions $F^{(i)}$ give two families of solutions which are obtained by expressing~$F^{(i)}_\lambda$ in terms of $\Psi^{(a,b)}_\lambda$ and attaching an index $+$ (respectively $-$) to them. For $\lambda,\varepsilon>0$, the set of solutions which decay at infinity is
\begin{gather*}\label{solutions}
G^{(1)}_{\lambda} = ({\rm i}\lambda)^{1/2} \Psi^{(0,0)}_{\lambda,-} e_1, \\
G^{(2)}_{\lambda} = \frac12 ({\rm i}\lambda)^{3/2} \Psi^{(\frac12,-\frac12)}_{\lambda,-} e_2 +
\sqrt{\frac{{\rm i}}{4\lambda}} \Psi^{(0,0)}_{\lambda+\frac12,-} e_1 + \sqrt{\frac{{\rm i}}{4\lambda}}
\Psi^{(0,0)}_{\lambda-\frac12,-} e_1 , \\
G^{(3)}_{\lambda} = \frac12 ({\rm i}\lambda)^{1/2} \Psi^{(\frac12,\frac12)}_{\lambda,-} e_3 -
\sqrt{\frac{{\rm i}}{4\lambda}} \Psi^{(0,0)}_{\lambda+\frac12,-} e_1 + \sqrt{\frac{{\rm i}}{4\lambda}}
\Psi^{(0,0)}_{\lambda-\frac12,-} e_1, \\
G^{(4)}_{\lambda} = \frac12 ({\rm i}\lambda)^{1/2} \Psi^{(\frac12,\frac12)}_{\lambda,-} e_4 +
\sqrt{\frac{{\rm i}}{4\lambda}} \Psi^{(0,0)}_{\lambda+\frac12,-} e_1 - \sqrt{\frac{{\rm i}}{4\lambda}}
\Psi^{(0,0)}_{\lambda-\frac12,-} e_1,\\
G^{(5)}_{\lambda} = \frac12 ({\rm i}\lambda)^{3/2} \Psi^{(\frac12,-\frac12)}_{\lambda,0} e_5 +
\sqrt{\frac{{\rm i}}{4\lambda}} \Psi^{(0,0)}_{\lambda+\frac12,-} e_1 + \sqrt{\frac{{\rm i}}{4\lambda}}
\Psi^{(0,0)}_{\lambda-\frac12,+} e_1, \\
G^{(6)}_{\lambda} = ({\rm i}\lambda)^{1/2}\Psi^{(0,0)}_{\lambda,-} e_6 + \frac{1}{4\sqrt{\lambda}} \left({\rm i}^{\frac32}\left(\lambda+\frac12\right)\Psi^{(\frac12,-\frac12)}_{\lambda+\frac12,-} (e_2+e_5)
+ {\rm i}^{\frac12} \Psi^{(\frac12,\frac12)}_{\lambda+\frac12,-} (e_3-e_4)\right) \\
\hphantom{G^{(6)}_{\lambda} =}{}
-\frac{1}{4\sqrt{\lambda}} \left({\rm i}^{\frac32}\left(\lambda-\frac12\right)
\Psi^{(\frac12,-\frac12)}_{\lambda-\frac12,-} (-e_2-e_5) + {\rm i}^{\frac12}
\Psi^{(\frac12,\frac12)}_{\lambda-\frac12,-}(e_3-e_4)\right) + \frac{(i\lambda)^\frac{1}{2}}{\lambda^2-\frac14}
\Psi^{(0,0)}_{\lambda,-} e_1.
\end{gather*}
The functions $G^{(i)}$ are the physical conformal blocks.

\section{Conclusions and outlook}\label{section6}

Interplay of ideas from representation theory, integrable systems and conformal field theory has a~long history. Traditionally, these interrelations appeared mostly in the study of two-dimensional CFTs. Significant advances in this area of mathematics and mathematical physics were made by professor Tarasov and professor Varchenko, e.g., in their study of the Knizhnik--Zamolodchikov equations~\cite{Tarasov:1994yy, Tarasov:1993vs}.

Recent progress the conformal bootstrap programme provides strong motivation to try and extend the methods of representation theory to higher dimensions. We hope to have convinced the reader that the harmonic analysis approach outlined in this text holds some promise in this endeavour.

The present review was focused on supersymmetric theories. In such setups, there are still many open questions amenable to our methods. While in this work all constructions were illustrated on a comparatively simple example in one dimension, explicit computations of a similar type in four dimensions have been done in~\cite{N1D4_paper}.

Here we did not discuss at length the role played by integrability. Some aspects of it can be found in~\cite{Isachenkov:2017qgn}. On the other hand, it should be clear that the methods presented here extend to other contexts such as defect CFTs. Some investigations in this direction appeared already in~\cite{Isachenkov:2018pef}, but more work is in progress which seems to provide a non-trivial and highly structured modification of the theory described in the text. Finally, the spaces of $K$-spherical functions have been studied recently in~\cite{stokman2020npoint}, where the authors placed a particular emphasis on their representation through the Eisenstein integral. Investigation of analytic properties of these functions also seem to be a necessary requirement for analysis of crossing equations. We hope to turn to this question in the future.

\appendix

\section{Euclidean conformal group}\label{appendixA}

In this appendix we collect some details about the euclidean conformal group. Its subgroups, decompositions and representation that play a role in the main text are defined. For many more details, the reader is referred to \cite{Dobrev:1977qv}.

The group of conformal transformations of the euclidean space $\mathbb{R}^d$ is $G={\rm SO}(d+1,1)$. Its Lie algebra $\mathfrak{g}$ is spanned by generators of translations $P_\mu$, rotations $M_{\mu\nu}$, dilations $D$ and special conformal transformations $K_\mu$, obeying the following non-vanishing brackets
\begin{gather*}
 [D,P_\mu] = P_\mu, \quad [D,K_\mu] = - K_\mu, \qquad [D,M_{\mu\nu}] = 0,\\
 [M_\mu\nu,P_\rho] = \delta_{\nu\rho} P_\mu - \delta_{\mu\rho} P_\nu, \qquad [M_{\mu\nu},K_\rho] = \delta_{\nu\rho} K_\mu - \delta_{\mu\rho} K_\nu,\\
 [M_{\mu\nu},M_{\rho\sigma}] = \delta_{\nu\rho} M_{\mu\sigma} - \delta_{\mu\rho} M_{\nu\sigma} + \delta_{\nu\sigma} M_{\rho\mu} - \delta_{\mu\sigma} M_{\rho\nu},\\
 [K_\mu,P_\nu] = 2 (M_{\mu\nu} - \delta_{\mu\nu}D) .
\end{gather*}
The Lie algebra $\mathfrak{g}$ is graded with respect to the eigenvalue under the adjoint action of the dilation generator, called conformal weight. Thus we can write
\begin{equation*}
 \mathfrak{g} = \mathfrak{g}_{-1} \oplus \mathfrak{g}_0 \oplus \mathfrak{g}_1 = \operatorname{span}\{K_\mu\} \oplus \operatorname{span}\{D,M_{\mu\nu}\} \oplus \operatorname{span}\{P_\mu\}.
\end{equation*}
Subgroups of $G$ that correspond to subalgebras of degree zero, and non-positive degree play an important role in the representation theory of $G$. We denote them by
\begin{equation*}
 K = {\rm SO}(1,1) \times {\rm SO}(d), \quad P = ({\rm SO}(1,1) \times {\rm SO}(d)) \ltimes \mathbb{R}^d .
\end{equation*}
The later subgroup, generated by dilations, rotations and special conformal transformations is a parabolic subgroup of~$G$. Since it is the only parabolic subgroup that plays a role in our considerations, we shall call it the parabolic subgroup. The quotient of~$G$ by $P$ is the compactified euclidean space
\begin{equation*}
 G/P\cong S^d.
\end{equation*}

\subsection{Unitary irreducible representations}\label{appendixA.1}

All unitary irreducible representations of the conformal groups can be constructed from {\it elementary representations} by taking subrepresentations and quotients.

A finite dimensional irreducible representation of the group $K$ is specified by a conformal weight~$\Delta$ and the highest weight $\mu$ of the rotation group ${\rm SO}(d)$. Let $\rho_{\Delta,\mu}$ be a representation of the parabolic subgroup which extends this representation of $K$ by requiring that special conformal transformations act trivially. An elementary representation $\pi_{\Delta,\mu}=[\Delta,\mu]$ is a representation of~$G$ induced from $\rho_{\Delta,\mu}$
\begin{equation*}
 \pi_{\Delta,\mu} = \operatorname{Ind}_{P}^{G}\rho_{\Delta,\mu}.
\end{equation*}
More explicitly, the carrier space of $\pi$ is that of right-covariant functions
\begin{equation*}
 \Gamma_{\Delta,\mu} = \big\{ f \colon G \xrightarrow{} V\, |\, f(gp) = \rho_{\Delta,\mu}(p)^{-1} f(g)\big\} .
\end{equation*}
The space $V$ is the carrier space of the inducing representation $\rho$. The action on $\Gamma$ is given by multiplication of the argument of a functions $(g\cdot f)(g') = f(gg')$. The representation $\pi_{\Delta\mu}$ is said to be {\it of type I} if $\mu=(0,\dots,0,l)$ is a symmetric traceless tensor.

Elementary representations are generically irreducible but not unitary. Unitary representations belong to either principal, complementary or discrete series. The principal series have
\begin{equation*}
 \Delta \in \frac{d}{2}+ {\rm i}\mathbb{R},\qquad \mu\ \text{arbitrary} .
\end{equation*}
They are unitary with respect to the inner product
\begin{equation}
 (f_1,f_2) = \int_N {\rm d}x\, \langle \bar f_1,f_2\rangle, \label{p1}
\end{equation}
where integration is over a section of $P$-orbits. This is well-defined (independent of the choice of section) if and only if $\Delta + \bar\Delta = d$, which leads to the above restriction on the conformal dimension.

For $\Delta \notin d/2+{\rm i}\mathbb{R}$, the inner product~(\ref{p1}) is not well-defined. However, in some cases, there exist other invariant scalar products which make the elementary representations unitary. These representations are said to constitute the {\it complementary series}. For type one representations, we have
\begin{equation*}
 l=0,\qquad 0<\Delta<d,\qquad l>0,\qquad 1<\Delta<d-1.
\end{equation*}
Complementary series representations can be obtained by analytic continuation of discrete series of $\widetilde{{\rm SO}}(d,2)$.

Discrete series representations are defined by the condition that their matrix coefficients are square-integrable functions on the group. They are not elementary, but rather subquotients of elementary representations. As indicated by their name, discrete series representations have $\Delta = d/2 + n$, $n\in \mathbb{N}$. These representation only exist when $d$ is odd.

\section{Superconformal algebras of type I}\label{appendixB}

In this appendix we define what is meant by a superconformal algebra and introduce types I and II. While some of the discussion of the main text applies equally well to both types, the construction of Casimir equations relies on the algebra being of type I.

Let $\mathfrak{g} = \mathfrak{g}_{(0)} \oplus \mathfrak{g}_{(1)}$ be a finite-dimensional Lie superalgebra. We say that $\mathfrak{g}$ is a superconformal algebra if its even part $\mathfrak{g}_{(0)}$ contains the conformal Lie algebra $\mathfrak{so}(d+1,1)$ as a direct summand and the odd part $\mathfrak{g}_{(1)}$ decomposes as a direct sum of spinor representations of $\mathfrak{so}(d)\subset\mathfrak{so}(d+1,1)$ under the adjoint action.

If this is the case, we denote the dilation generator of the bosonic conformal Lie algebra by~$D$. Eigenvalues with respect to $\text{ad}_D$ give a decomposition of $\mathfrak{g}$ into the sum of eigenspaces
\begin{equation*}
\mathfrak{g} = \mathfrak{g}_{-1}\oplus\mathfrak{g}_{-1/2} \oplus \mathfrak{g}_0\oplus\mathfrak{g}_{1/2}\oplus\mathfrak{g}_{1} = \mathfrak{g}_{-1} \oplus \mathfrak{s} \oplus \mathfrak{k} \oplus \mathfrak{q}\oplus \mathfrak{g}_{1} .
\end{equation*}
The even part of $\mathfrak{g}$ is composed of $\mathfrak{g}_{\pm 1}$ and $\mathfrak{k}$ where $\mathfrak{g}_{-1}={\mathfrak{n}}$ contains the generators $K_\mu$ of special conformal
transformations while $\mathfrak{g}_{1} = \mathfrak{n}$ is spanned by translations $P_\mu$. Dilations, rotations and internal symmetries make up
\[ \mathfrak{k} = \mathfrak{so}(1,1) \oplus \mathfrak{so}(d) \oplus \mathfrak{u}.\] Generators of $\mathfrak{g}_{\pm1/2}$, are supertranslations~$Q_\alpha$ and super special conformal transformations $S_\alpha$. We shall also denote these summands as $\mathfrak{s}
= \mathfrak{g}_{-1/2}$ and $\mathfrak{q} = \mathfrak{g}_{1/2}$. All elements of non-positive degree make up a subalgebra $\mathfrak{p}$ of $\mathfrak{g}$ that will be referred to as the parabolic subalgebra
\begin{equation*}
 \mathfrak{p} = \mathfrak{g}_{-1} \oplus \mathfrak{g}_{-1/2} \oplus \mathfrak{g}_0 .
\end{equation*}
There is a unique (connected) corresponding subgroup $P\subset G$ such that $\mathfrak{p} = {\rm Lie}(P)$. The superspace can be identified with the supergroup of translations and supertranslations. It is defined as the homogeneous space $M = G/P$.

The above structure is present in any superconformal algebra. In this work, we shall mainly consider those $\mathfrak{g}$ which satisfy an additional condition of being of type~I. This means that the
odd subspace decomposes as a direct sum of two irreducible representations of $\mathfrak{g}_{(0)}$ under the adjoint action
\begin{equation*}
\mathfrak{g}_{(1)} = \mathfrak{g}_+ \oplus \mathfrak{g}_- .
\end{equation*}
The two modules $\mathfrak{g}_\pm$ are then necessarily dual to each other and further satisfy
\begin{equation*}
 \{\mathfrak{g}_\pm , \mathfrak{g}_\pm\} = 0.
\end{equation*}
In addition, the bosonic algebra assumes the form
\begin{equation}\label{U1R}
\mathfrak{g}_{(0)} = [\mathfrak{g}_{(0)},\mathfrak{g}_{(0)}] \oplus \mathfrak{u}(1) .
\end{equation}
The $\mathfrak{u}(1)$ summand is a part of the internal symmetry algebra. Its generator will be denoted by~$R$. All elements in $\mathfrak{g}_+$ possess the same $R$-charge. The same is true for the
elements of $\mathfrak{g}_-$, but the $R$-charge of these elements has the opposite value. Elements in the even subalgebra $\mathfrak{g}_{(0)}$, on the other hand, commute with~$R$.

Let us denote the intersections of the subspaces $\mathfrak{q}$ and $\mathfrak{s}$ with $\mathfrak{g}_\pm$ by
\begin{equation*} 
\mathfrak{q}_\pm = \mathfrak{q}\cap\mathfrak{g}_{\pm} ,\qquad \mathfrak{s}_\pm =\mathfrak{s}\cap\mathfrak{g}_\pm .
\end{equation*}
The subspaces $\mathfrak{q}_\pm$ and $\mathfrak{s}_\pm$ do not carry a representation of $\mathfrak{g}_{(0)}$, but they do carry a representation of~$\mathfrak{k}$. This also
means that in type I superconformal algebras, the action of~$\mathfrak{k}$ on super-translations decomposes into two or more irreducible representations. It turns out that
\begin{equation*}
 \dim (\mathfrak{q}_\pm) = \dim (\mathfrak{s}_\pm) = \dim (\mathfrak{g}_{(1)})/4 .
\end{equation*}
The full list of type I superconformal algebras, which follows directly from Kac's classification~\cite{Kac:1977em}, is
\begin{equation*}
 \mathfrak{sl}(2|\mathcal{N}),\quad \mathfrak{sl}(2|\mathcal{N}_1)\oplus\mathfrak{sl}(2|\mathcal{N}_2),\quad \mathfrak{psl}(2|2),\quad \mathfrak{osp}(2|4),\quad \mathfrak{sl}(4|\mathcal{N}),\quad \mathfrak{psl}(4|4) .
\end{equation*}
The presented list is that of complexified Lie superalgebras~-- for different spacetime signatures one considers their various real forms.

\section{Induced and coinduced representations}\label{appendixC}

In this appendix we collect some properties of the two types of representations that play a role in the main text, following Blattner~\cite{Blattner}. These representations are obtained by processes of induction and coinduction.

Given any algebra $\mathcal{A}$, a subalgebra $\mathcal{B}$ and a representation $\rho\colon \mathcal{B}\xrightarrow{}\text{End}(W)$ of $\mathcal{B}$, we can define two representations of $\mathcal{A}$ on the following spaces
\begin{equation*}
 \text{Ind}_{\mathcal{B}}^{\mathcal{A}}\rho = \mathcal{A} \otimes_{\mathcal{B}} W, \qquad \text{Coind}_{\mathcal{B}}^{\mathcal{A}}\rho = \text{Hom}_{\mathcal{B}}(\mathcal{A},W),
\end{equation*}
Elements of the first space are linear combinations of vectors $a\otimes w$, under identifications
\begin{equation*}
 ab \otimes w \sim a\otimes bw, \qquad a\in\mathcal{A},\qquad b\in\mathcal{B},\qquad w\in W,
\end{equation*}
and the action of $\mathcal{A}$ is the left regular one. In the second space, elements are $\mathcal{B}$-equivariant maps
\begin{equation*}
 \varphi\colon \ \mathcal{A}\xrightarrow{}W,\qquad \varphi(ba) = b\varphi(a),
\end{equation*}
and the action now is $(a\varphi)(a') = \varphi(a'a)$. The two modules introduced are called induced and coinduced modules, respectively. We defied them as left $\mathcal{A}$-modules. For an arbitrary algebra, induced and coinduced modules are formally related by duality. We shall now explain this relation in the context of representations of Lie groups and Lie algebras.

When studying representations of groups and Lie algebras, one can replace these algebraic objects by associative algebras that have the same representation theory. For groups, this is the group algebra (algebra of functions on the group under convolution) and for Lie algebras, it is the universal enveloping algebra. Thus, the above constructions give definitions of induction and coinduction for groups and algebras. For example, if~$G$ is any group, $H\subset G$ a~subgroup and~$\rho$ a~representation of~$H$ on the space~$W$, we put $\mathcal{A} = \mathbb{C}[G]$ and $\mathcal{B} = \mathbb{C}[H]$. Thus, the induced module of $\mathbb{C}[G]$ (and thereby $G$) is
\begin{equation*}
 \text{Ind}_H^G W = \mathbb{C}[G]\otimes_{\mathbb{C}[H]}W,
\end{equation*}
with the regular left action. Similar comment apply for coinduced representations. If $G$ is a~Lie group, one may equivalently view the above module as the space of covariant vector-valued functions on the group
\begin{equation*}
 \Gamma = \big\{ f \colon G\xrightarrow{}W\, |\, \varphi\big(gh^{-1}\big)=\rho(h)\varphi(g)\big\} ,
\end{equation*}
under the left-regular action $(g\cdot f)(x)=f\big(g^{-1}x\big)$. This view is more in line with our discussion in the main text.

There is a close relation between induced representations of Lie groups and coinduced representations of their Lie algebras that we shall now explain. Let $G$ be a Lie group, $H\leq G$ a Lie subgroup and $\mathfrak{g}={\rm Lie}(G)$, $\mathfrak{h}={\rm Lie}(H)$. Let $W$ be a finite dimensional representation of $H$ and use the same letter for the derivative representation of $\mathfrak{h}$. Then
\begin{equation}\label{ind-coind-1}
 d(\text{Ind}_H^G W) = \text{Coind}_{\mathfrak{h}}^{\mathfrak{g}} W = \text{Hom}_{U(\mathfrak{h})}(U(\mathfrak{g}),V).
\end{equation}
To see how this comes about, recall that the representation space on the right hand side consists of linear maps $U(\mathfrak{g})\xrightarrow{}W$ which commute with the action of $U(\mathfrak{h})$ on $U(\mathfrak{g})$ (by left multiplication) and $W$. The action of $x\in\mathfrak{g}$ on such a map is given by
\begin{equation*}
 (x\psi)(A) = \psi(A x),\qquad A\in U(\mathfrak{g}).
\end{equation*}
To see how~(\ref{ind-coind-1}) comes about consider an analytic function $f\colon G\xrightarrow{}W$. This function defines a linear map on the universal enveloping algebra through its Taylor coefficients
\begin{equation*}
 \psi\colon \ U(\mathfrak{g})\xrightarrow{}W, \qquad \psi(A) = \mathcal{R}_A f (e) .
\end{equation*}
Here $\mathfrak{R}_A$ is a differential operator corresponding to the element $A$ of the universal enveloping algebra, constructed out of right-invariant vector fields. Conversely, the knowledge of all Taylor coefficients can be used to recover~$f$. Covariance properties of $\psi$ follow from those of~$f$.

We mentioned that there is a formal relation of duality between induced and coinduced representation of arbitrary algebras. For Lie algebras, the duality takes a concrete form
\begin{equation*}
 \text{Coind}_\mathfrak{h}^\mathfrak{g}(W^*) = \big(\text{Ind}_\mathfrak{h}^\mathfrak{g}W\big)^* .
\end{equation*}
To see that this is true, let $V = \text{Ind}_\mathfrak{p}^\mathfrak{g} W $. Given $\psi\in V^*$ and $A\in U(\mathfrak{g})$ define the function
\begin{equation*}
 \psi\colon \ U(\mathfrak{g})\xrightarrow{} W^*,\qquad \psi(A)(w) = f(s(A)\otimes w) .
\end{equation*}
where $s$ is the antipode in $U(\mathfrak{g})$. It is clear that $\psi(A)$ is an element of $W^*$ and that $\psi$ is a~linear map. It also belongs to the coinduced module $\pi = \text{Coind}_\mathfrak{p}^\mathfrak{g} W^*$. This follows from the computation
\begin{equation*}
 \psi(BA)(w) = f (s(A)s(B) w) = \psi(A)(\sigma(B) w) = \big(B(\psi(A))\big)(w) .
\end{equation*}
Here, $B$ is an element of $U(\mathfrak{p})$. The last step uses the definition of the dual representation for the algebra $U(\mathfrak{p})$. The map $f\xrightarrow{}\psi$ is clearly linear. It also commutes with the action of~$U(\mathfrak{g})$. To see this, let $C\in U(\mathfrak{g})$. Then
\begin{align*}
 \widehat{(Cf)} (A) (w) & = (Cf)(\sigma(A)w) = f(\sigma(C)\sigma(A) w) = f(\sigma(AC) w)\\
 & = \psi(AC)(w) = (C\psi)(A)(w) .
\end{align*}
It is a simple matter to show that $f\xrightarrow{}\psi$ is a bijection. Therefore, the map establishes an isomorphism between the coinduced representation from~$W^\ast$ and the dual of the induced representation from~$W$.

In the context of conformal field theory, the states of the Hilbert space belong to representations induced from a parabolic subalgebra of the conformal Lie algebra. These representations are known as the parabolic Verma modules. Their dual modules form the algebraic principal series of representations. Algebraic principal series are naturally realised as coinduced representations. Their name steams the fact that the space of smooth vectors in a principal series representation of the conformal group~$G$ forms the algebraic principal series representation of the Lie algebra $\mathfrak{g}={\rm Lie}(G)$.

In the supersymmetric case, treating the induced representations of a superconformal group can be rather delicate. Therefore, for some purposes, including the analysis of tensor products of principal series representations, it is most convenient to work with the coinduced representations of the superconformal algebra, as was done in~\cite{Buric:2019rms}.

\section{Elements of supergeometry}\label{appendixD}

In this appendix we collect some properties of supermanifolds and Lie supergroups, following the classical work of Kostant \cite{Kostant:1975qe}. We hope these may be useful to some readers by offering a~way to put constructions of Sections~\ref{section3}--\ref{section5} on a firm mathematical basis.

Recall that, by definition, a supermanifold $M$ is a topological space $X$ together with a~sheaf~$A$ of superalgebras, such that around any point $x\in X$ there is an open neighbourhood $U$ with $A(U)\cong C^\infty(U)\otimes\Lambda_n$, where $\Lambda_n$ is the Grassmann algebra on $n$ generators. The number $n$ is called the odd dimension of~$M$. For any open set $V\subset X$, $A(V)$ is a commutative superalgebra. It is a non-trivial, but familiar, fact that the supermanifold can be completely recovered from its structure algebra~$A(X)$.

Some constructions regarding supermanifolds are more easily formulated in terms of a certain coalgebra $A(X)^\ast$ rather than the structure algebra itself. The $A(X)^\ast$ is defined as the space of all elements in the full dual $A(X)'$ which vanish on some ideal of finite codimension in $A(X)$. Elements of $A(X)^\ast$ are referred to as {\it distributions with finite support}. One observes that $A(X)^\ast$ is a supercocommutative coalgebra. Namely, let $i$ and $\Delta$ be the natural injection and the diagonal map
\begin{gather*}
 i \colon \ A(X)'\otimes A(X)'\xrightarrow{} (A(X)\otimes A(X))',\qquad i(v\otimes w)(f\otimes g) = (-1)^{|w| |f|} v(f) w(g),\\
 \Delta \colon \ A(X)'\xrightarrow{}(A(X)\otimes A(X))',\qquad (\Delta v)(f\otimes g) = v (f g),\qquad v,w\in A(X)',\ f,g\in A(X).
\end{gather*}
Then one can show $\Delta(A(X)^\ast)\subset A(X)^\ast\otimes A(X)^\ast$, so the diagonal map makes $A(X)^\ast$ into a~coalgebra. One again has that $A(X)^\ast$ determines the sheaf $A$. For example, $X$ as a set can be recovered either as the set of all homomorphisms $A(X)\rightarrow\mathbb{R}$, or as the set of all group-like elements in $A(X)^\ast$. The coalgebra $A(X)^\ast$ also plays a prominent role in the theory of Lie supergroups and their actions on supermanifolds, as will be outlined presently.

\subsection{Lie supergroups}\label{appendixD.1}

Let $\mathfrak{g}$ be a Lie superalgebra, $H$ a group and $\pi\colon H\xrightarrow{}\text{Aut}(U(\mathfrak{g}))$ a representation of $H$ by algebra automorphisms. Further, write $F(H)$ for the group algebra of $H$. The {\it smash product} $E(H,\mathfrak{g},\pi)$ is a supercocommutative Hopf algebra constructed as follows:
\begin{enumerate}\itemsep=0pt
\item[1)] As a vector space $E = F(H)\otimes U(\mathfrak{g})$.
\item[2)] The multiplication in $F(H)$ and $U(\mathfrak{g})$ is defined in the usual way and $h x h^{-1} = \pi(h) x$.
\item[3)] The comultiplication $\Delta$, counit $\eta$ and the antipode $\sigma$ are defined as
\begin{gather*}
 \Delta(h) = h\otimes h,\qquad \Delta(x) = 1\otimes x + x\otimes 1,\qquad \eta(h)=1,\qquad \eta(x)=0,\\
 \sigma(h) = h^{-1},\qquad \sigma(x) = - x,\qquad \sigma(A B) = (-1)^{|A||B|}\sigma(B)\sigma(A) .
\end{gather*}
\end{enumerate}
In these formulas $h\in H$, $x\in \mathfrak{g}$ and $A,B\in U(\mathfrak{g})$ are arbitrary. The set of group-like elements of $E$ is precisely $H$ and that of primitive elements is $\mathfrak{g}$. Here $\mathfrak{g}$ is identified with a subspace of $U(\mathfrak{g})$ in the obvious way. Conversely, given a supercocommutative Hopf algebra $E$ with the group of group-like elements $H$ and the Lie superalgebra of primitive elements $\mathfrak{g}$ one can show that a representation $\pi$ exists such that $E=E(H,\mathfrak{g},\pi)$. Now assume that $\mathfrak{g} = \mathfrak{g}_{\bO}\oplus\mathfrak{g}_{\b1}$ is a Lie superalgebra and $G_0$ the connected, simply connected Lie group whose Lie algebra is $\mathfrak{g}_{\bO}$. Then there is a unique representation $\pi$ on $\mathfrak{g}$ by Lie superalgebra automorphisms which reduces to the adjoint representation on $\mathfrak{g}_{\bO}$. The smash product $E(G_0,\mathfrak{g},\pi)$ is called the simply-connected Lie--Hopf algebra associated with $\mathfrak{g}$ and denoted by $E(\mathfrak{g})$.

A supermanifold $(X,A)$ is said to be a Lie supergroup if the coalgebra $A(X)^\ast$ is a Hopf algebra. By the above remarks, in this case $A(X)^\ast$ is a smash product $E(G_{0},\mathfrak{g},\pi)$ with $X=G_0$. In fact, if $X$ is simply connected, it can be shown that $A(X)^\ast = E(\mathfrak{g})$ for some Lie superalgebra, called the Lie superalgebra of $(X,A)$.

\subsection{Supergroup actions}\label{appendixD.2}

Assume now that $G=(G_0,A)$ is a Lie supergroup and $M=(Y,B)$ another supermanifold. We will say that $G$ acts on $M$ if there is a map $A(G_0)^\ast \otimes B(Y)^\ast \xrightarrow{} B(Y)^\ast$, $u\otimes w\mapsto u\cdot w$, which satisfies
\begin{equation*}
 \Delta u = \sum_i u_i'\otimes u_i'',\qquad \Delta w = \sum_j w_j'\otimes w_j''\implies \Delta(u\cdot w) = \sum_{i,j} (-1)^{|u_i''||w_j'|} u_i'\cdot w_j'\otimes u_i''\cdot w_j''.
\end{equation*}
In this case, the structure algebra $B(Y)$ is a $A(G_0)^\ast$-module through
\begin{equation*}
 \pi \colon \ A(G_0)^\ast\xrightarrow{}\text{End}(B(Y)),\qquad \langle w, \pi(u) f\rangle = (-1)^{|u| |w|}\langle \sigma(u)\cdot w, f\rangle .
\end{equation*}
The later is called the coaction representation of $G$. The action of $G$ is fully determined by the corresponding coaction representation. Bearing in mind that $A(G_0)^\ast = E(\mathfrak{g})$, we see that a Lie supergroup action can be though of as a pair of representations of the underlying group~$G_0$ and of Lie superalgebra~$\mathfrak{g}$ on the vector space $B(Y)$, which satisfy a compatibility condition.

Dually, there is a map $\varphi\colon B(Y)\xrightarrow{}B(Y)\otimes A(G_0)$ that makes $B(Y)$ into a comodule-algebra of $A(G_0)$. This means that $\varphi$ is a morphism of algebras which is compatible with the Hopf algebra structure of $A(G_0)$. For example, $\varphi$ satisfies
\begin{equation*}
 (1\otimes\Delta)\circ\varphi = (\varphi\otimes1)\circ\varphi\colon \ B(Y)\xrightarrow{}B(Y)\otimes A(G_0)\otimes A(G_0),
\end{equation*}
along with a number of other compatibility conditions, see, e.g., \cite{madore_1999}. Let $p$ be a point in $G_0$, considered as a morphism $p\colon A(G_0)\xrightarrow{}\mathbb{R}$. Then one can form the map $(1\otimes p)\circ\varphi\colon B(Y)\xrightarrow{}B(Y)$. For obvious reasons, we refer to such compositions with~$p$ as {\it evaluations}. Running over all points~$p$, we get a representation of the~$G_0$ on $B(Y)$. This agrees with the coaction representation~$\pi$ from above. The evaluated action of the bosonic group~$G_0$ and the infinitesimal action of the conformal Lie superalgebra fit together to form a representation of the Lie--Hopf algebra $A(G_0)^\ast$ on $B(Y) = \mathcal{M}$.

\subsection*{Acknowledgements} This work reviewed results obtained in collaboration with Volker Schomerus and Zhenya Sobko, some of which originated in the work with Misha Isachenkov. I warmly thank all my collaborators. I am also indebted to Maja Buri\'c and Aleix Gimenez-Grau for a number of illuminating discussions and to anonymous referees for suggesting a number of improvements of the text. Finally, I acknowledge the support by the Deutsche Forschungsgemeinschaft (DFG, German Research Foundation) under Germany’s
Excellence Strategy~-- EXC 2121 ``Quantum Universe''~-- 390833306.

\pdfbookmark[1]{References}{ref}
\LastPageEnding

\end{document}